\RequirePackage{fix-cm}
\documentclass[smallextended]{svjour3}
\smartqed  
\journalname{International Journal of Parallel Programming}

\usepackage{amssymb,listings,color,enumitem,eqlist}
\usepackage[ruled]{algorithm2e}
\usepackage{graphicx}
\usepackage{url}
\usepackage{draftwatermark}
\SetWatermarkText{DRAFT 17.06.2014 -- Please don't distribute.}
\SetWatermarkScale{0.25}
\SetWatermarkLightness{0.95}

\newcommand{\pocl}[0]{\textit{pocl} }

\newcommand{\lb}{[}

\sloppypar

\begin{document}

\markboth{P. J\"a\"askel\"ainen et al.}{pocl: A Performance-Portable OpenCL Implementation}
\title{pocl: A Performance-Portable OpenCL Implementation}

\author{Pekka~J\"a\"askel\"ainen \and
Carlos~S\'anchez~de~La~Lama \and
Erik~Schnetter \and
Kalle~Raiskila \and
Jarmo~Takala \and
Heikki~Berg
}




\institute {P. J\"a\"askel\"ainen and J. Takala \at
Tampere University of Technology, Finland \newline
\{pekka.jaaskelainen,jarmo.takala\}@tut.fi
\and
C.S. de La Lama \at
Knowledge Development for POF, Madrid, Spain \newline
carlos.sanchez@kdpof.com
\and
E. Schnetter \at
Perimeter Institute for Theoretical Physics, Canada; \newline
Department of Physics, University of Guelph, Guelph, Canada; \newline
Center for Computation \& Technology, Louisiana State University, USA \newline
eschnetter@perimeterinstitute.ca
\and
K. Raiskila and H. Berg \at
Nokia Research Center, Finland \newline
kalle.raiskila@nokia.com
}

\date{Received: date / Accepted: date}
\maketitle

\begin{abstract}
 
OpenCL is a standard for parallel programming of heterogeneous systems. The benefits 
of a common programming standard are clear; multiple vendors can provide support for 
application descriptions written according to the standard, thus reducing the 
program porting effort. While the standard brings the obvious benefits of platform portability,
the performance portability aspects are largely left to the programmer.
The situation is made worse due to multiple proprietary vendor implementations with
different characteristics, and, thus, required optimization strategies.

In this paper, we propose an OpenCL implementation that is both 
portable and performance portable. At its core is a kernel 
compiler that can be used to exploit the data parallelism of OpenCL 
programs on multiple platforms with different parallel hardware styles. 
The kernel compiler is modularized to perform target-independent 
parallel region formation separately from the target-specific parallel mapping of 
the regions to enable support for various styles of fine-grained parallel 
resources such as subword SIMD extensions, SIMD datapaths and static multi-issue. 
Unlike previous similar techniques that work on the source level,
the parallel region formation retains the information of the data parallelism 
using the LLVM IR and its metadata infrastructure. This data can be exploited 
by the later generic compiler passes for efficient parallelization.

The proposed open source implementation of OpenCL is also platform 
portable, enabling OpenCL on a wide range of architectures, both already commercialized and 
on those that are still under research. The paper describes how the portability of 
the implementation is achieved. We test the two aspects to portability by utilizing 
the kernel compiler and the OpenCL implementation to run OpenCL applications 
in various platforms with different style of parallel resources.
The results show that most of the benchmarked applications when compiled using 
pocl were faster or close to as fast as the best proprietary OpenCL implementation 
for the platform at hand.



\end{abstract}

\keywords{OpenCL, LLVM, GPGPU, VLIW, SIMD, Parallel Programming, Heterogeneous Platforms, Performance Portability}


\section{Introduction}

Widely adopted programming standards help the programmer by reducing the 
porting effort when moving the software to a new platform. 
\textit{Open Computing Language (OpenCL)}~\cite{OpenCL12}  is a 
relatively new standard for parallel programming of heterogeneous 
systems. 
OpenCL allows the programmer to describe the program parallelism by
expressing the computation in the \textit{Single Program Multiple Data (SPMD)}
style. In this style, multiple parallel \textit{work-items} execute the
same kernel function in parallel with synchronization expressed explicitly
by the programmer. Another key concept in OpenCL is the \textit{work-group}
which collects a set of coupled \textit{work-items} possibly synchronizing
with each other. However, across multiple \textit{work-groups} executing the
same kernel there cannot be data dependencies. These concepts allow exploiting
parallelism in multiple levels for a single kernel description; 
inside a work-item, across work-items in a single work-group and across all
the work-groups in the \textit{work-space}.

While the OpenCL standard provides an extensive programming platform for portable 
heterogeneous parallel programming, the version 1.2 of the standard is quite
low-level, exposing a plenty of details of the platform to the programmer. Thus, 
using these platform queries, it is possible to adapt the program to each of
the platforms. However, this means that to achieve performance portability,
the programmer has to explicitly do the adaptation for each program separately. 
In addition, implementations of the OpenCL standard are vendor and
platform specific, thus acquiring the full performance of an OpenCL application
requires the programmer to become familiar with the special characteristics of
the implementation at hand and tune the program accordingly. This is a serious drawback for performance portability 
as manual optimizations are needed to port the same code to another platform.


In our earlier work~\cite{OpenCLASIP},  we used
kernel serialization for extracting instruction-level 
parallelism for a statically scheduled processor template with
the OpenCL vendor extension interface used for providing seamless
access to special function units. This initial approach proved to be
a good basis for a kernel compiler with improved support for performance 
portability.
In this article, we propose kernel compilation techniques that expose the implicit 
parallelism of OpenCL multiple work-item (WI) work-groups in a form that can 
be exploited in different types of parallel processing hardware.
We propose a set of compiler transformations, which together produce multi-WI work-group functions
that can be parallelized in multiple ways and using different granularities of parallel
resources, depending on the target. 
By separating the \textit{parallel region} extraction from the
actual parallel mapping of the multi-WI work-groups, we
obtain a basis for improving the performance portability of OpenCL kernels.
We have realized the proposed transformations as a modularized set of passes using
the LLVM compiler infrastructure~\cite{LLVM} and integrated them in an OpenCL implementation 
called \textit{Portable Computing Language (pocl)}. We have used \pocl to run applications 
on different processor architectures with different parallelism capabilities to test
the applicability of the proposed approach.

The main contributions of this article are as follows:
\begin{itemize}
 \item OpenCL kernel compilation techniques that separate the parallel region formation
from multiple WI work-group functions from the actual platform specific
parallelization methods;
 \item a kernel compiler that works on the LLVM IR, thus can support more 
kernel languages than OpenCL C via the {Standard Portable Intermediate Representation (SPIR)} 
standard~\cite{SPIR12}; and
 \item a complete OpenCL implementation, which allows performance portability over 
wide range of computing architectures with different styles and degrees of 
parallel hardware.
\end{itemize}

The remainder of the article is organized as follows.  The first section is
an overview to the OpenCL standard. It briefly describes its concepts that
are relevant to understanding the rest of the article. The higher level 
software architecture of our OpenCL implementation, \textit{pocl}, containing the 
proposed transformation passes is given in Section~\ref{sec:SoftwareArchitecture}. 
The actual kernel compiler techniques are described in
Sections~\ref{sec:KernelCompiler}.
Section~\ref{sec:MathFunctions} describes the vectorized mathematical functions used 
in \textit{pocl}.
Section~\ref{sec:Evaluation} evaluates the applicability of the proposed approach on several platforms and
compares the performance to the proprietary OpenCL implementations. 
Section~\ref{sec:RelatedWork} compares the proposed techniques 
to the related work. Finally, conclusions and the planned future work are 
presented in Section~\ref{sec:Conclusions}.

\section{Open Computing Language}


The \textit{OpenCL 1.2 framework} specifies three main parts: The \textit{OpenCL
Platform Layer} for querying information of the platform and the
supported devices, the \textit{OpenCL Runtime}  providing
programming interfaces for controlling the devices by queuing them
kernel execution and memory transfer commands, and the \textit{OpenCL
Compiler} for compiling the OpenCL C kernels for each targeted device.
OpenCL programs structure the computational parts of the application 
into \textit{kernels} defined in the OpenCL C kernel language, 
and specify that there shall be no data dependencies between the ``kernel 
instances'' (\textit{work-items}, analogous to loop iterations) by 
default. 

The example OpenCL C
kernel in Fig.~\ref{fig:opencl-c-vector-dot-product} 
can be executed on different width \textit{Single Instruction Multiple Data (SIMD)} 
hardware in parallel, parallelizing as many \textit{work-items} as there are parallel
processing elements in the device.
The call to \textit{get\_global\_id(0)} returns the index of the work-item in the
\textit{global index space} which in this case maps directly to the index in
the buffers.
A difference to standard C notation in this example
is the use of the \textit{global} qualifier in the kernel
arguments. This
is used to mark the pointers to point to buffers in \textit{global memory}.
Other disjoint explicitly addressed memory spaces in OpenCL C include 
the \textit{local memory} visible to single \textit{work-groups} (groups
of work-items within the global index space that can synchronize with
each other) at a time, the \textit{private memory} visible only to 
single \textit{work-items}, and the \textit{constant memory} for storing 
read-only data. The kernel can use vector datatypes, e.g. float4 is a four
element floating point vector.

\begin{figure}
\begin{minipage}[!t]{.3\linewidth}

\begin{lstlisting}[language=C,basicstyle=\footnotesize]
kernel void 
dot_product (global const float4 *a,  
             global const float4 *b, 
             global float *c) 
{ 
  int gid = get_global_id(0); 

  c[gid] = dot(a[gid], b[gid]); 
}
\end{lstlisting}

\end{minipage}

\caption{Vector dot product in OpenCL C.}

\label{fig:opencl-c-vector-dot-product}
\end{figure}

The OpenCL runtime API (a C language API) is used to launch kernels and 
data transfer commands in one or more \textit{compute devices} with \textit{event} 
synchronization. Thus, the targeted OpenCL platform consists of a host
device that executes the top level program and one or more devices that
perform the computation. Portability of OpenCL programs across a wide range of different 
heterogeneous platforms is achieved by describing the kernels as source code strings 
which are then explicitly compiled using the runtime API to the targeted devices. 
Thus, even if the runtime part of the application was distributed as a platform 
specific binary, the device programs can be recompiled to each device using 
the OpenCL C compiler of the platform at hand. 

%
%
In the kernel execution, the OpenCL programmer can describe 
two forms of parallelism:
work-item parallelism and work-group level parallelism. Parallelism within a
single work-item can be explicitly expressed using vector computations. In addition, the
implicit instruction level parallelism that can be described in traditional C functions
is also available: the programmer can define, e.g., for-loops inside work-items 
that can be parallelized by the compiler or the hardware, if there are no dependencies 
restricting the parallelization.

The important additional source of parallelism is the data level parallelism described by
multi-WI work-groups.
In OpenCL 1.2, multiple work-items in a work-group that execute an instance of
the same kernel described in the OpenCL C programming language can be assumed 
to be independent by default with only explicit synchronization constructs limiting
the parallelism~\cite{OpenCL12}. Thus, the device is free to execute 
the work-items in parallel or in serial manner. 
%
%
%
Task and thread level parallelism can be also exploited in OpenCL applications.
Multiple work-groups are assumed to be independent of each other, thus can be 
executed in parallel to exploit multiple hardware threads or multiple completely
independent cores. In addition, at the higher level of OpenCL applications, separate 
commands in an out-of-order command queue, and commands in different command queues
can be assumed to be independent of each other unless explicitly 
synchronized using events or command-queue barriers. 

\section{Portable OpenCL Implementation}
\label{sec:SoftwareArchitecture}

The proposed kernel compilation techniques are included in our OpenCL implementation, \textit{pocl}. 
We give an overview of its software architecture before going to the details of the transformations. 
The software architecture of \pocl is modularized to encourage code
reuse and to isolate the device specific aspects of OpenCL\@ to provide a
platform portable implementation.
The higher-level components of \pocl  are illustrated in
Fig.~\ref{fig:SoftwareComponents}. The implementation is divided
to parts that are executed in the host and to those that implement
device-specific behavior. The \textit{host layer} implementation is 
portable to targets with operating system C compiler support. 
The \textit{device layer} encapsulates
the operating system and \textit{instruction-set architecture (ISA)} 
specific parts such as code generation for the target device,
and orchestration of the execution of the kernels in the device.

\begin{figure}[t]
\begin{center}
\includegraphics[width=0.7\textwidth]{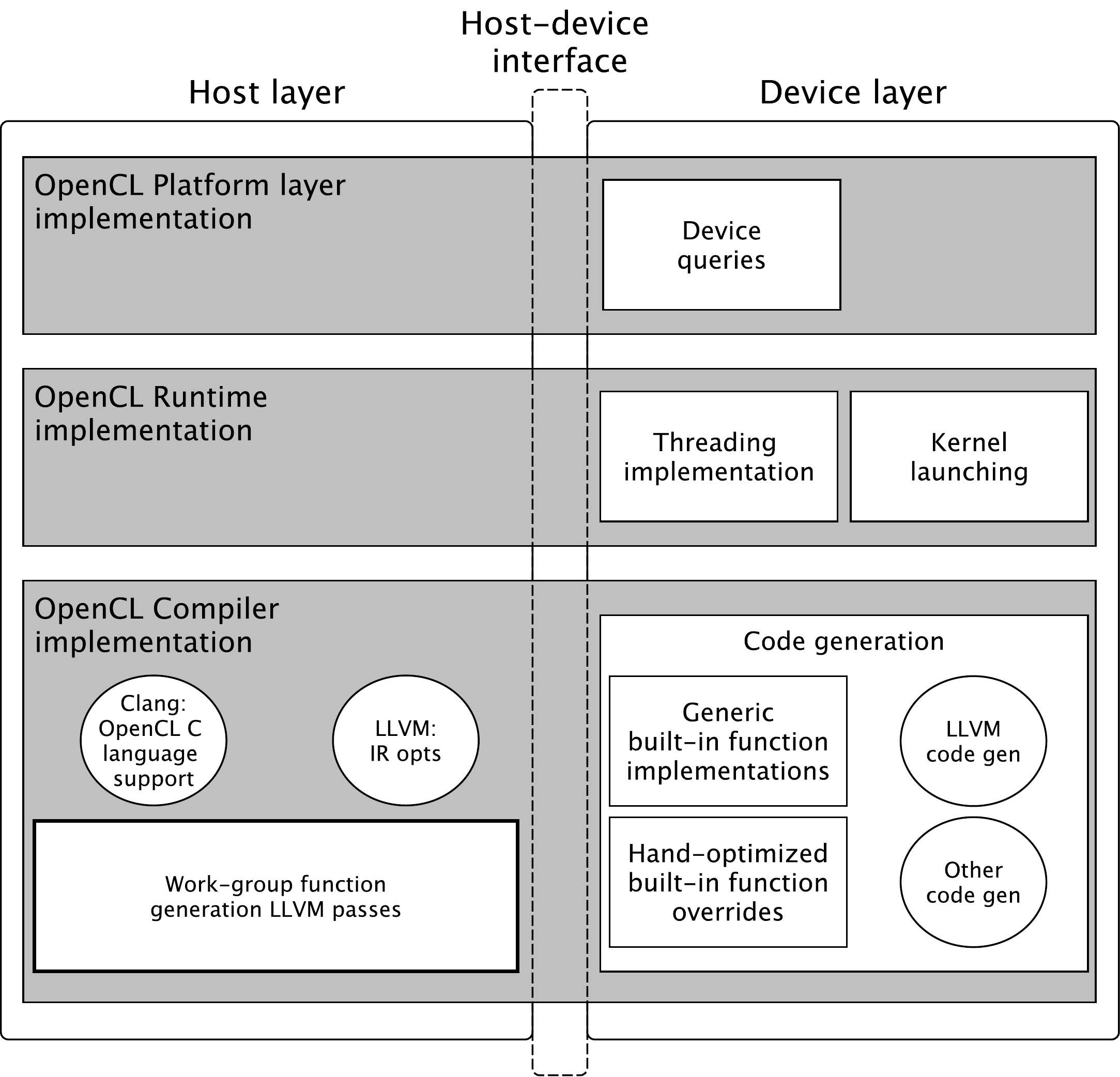}
\caption{The subcomponents in the OpenCL implementation. The host layer
includes parts that are executed in the OpenCL host. The device layer
is used as an hardware abstraction layer to encapsulate the device-specific
parts.}
\label{fig:SoftwareComponents}
\end{center}
\end{figure}

Most of the API implementations of the OpenCL framework in \pocl
are generic implementations written in C which call the \textit{device
layer} through a generic \textit{host-device interface} for device-specific 
parts. For example, when the OpenCL program queries for the number of
devices, \pocl returns a list of supported devices without
needing to do anything device-specific yet. However, when the application
asks for the size of the global memory in a device, the query is delegated
down to the device layer implementation of the device at hand.

The device layer consists of target-specific implementations for functionality such
as target-specific parts of the kernel compilation process, the final execution of 
the command queue including uploading the kernel to the device and launching it, 
querying device characteristics, etc.
%
The responsibilities between the device-specific and generic
parts in the currently supported device interfaces are as follows:  

\begin{description}
 \item[\texttt{basic}] A minimal example CPU device implementation.
The execution of kernels happens one work-group at a time without
multithreading. This driver can be used for implementing a device on
a POSIX-compliant operating system for the case where the host and the
device are the same. 

 \item[\texttt{pthread}] Similar to 'basic' except that it uses the POSIX threads~\cite{pthreads} 
library to execute multiple work-groups in parallel. This is an example
of a device layer implementation that is capable of exploiting the thread
level parallelism in multi-work-group execution.

 \item[\texttt{ttasim}] A proof-of-concept implementation of a simulated heterogeneous accelerator
setup. The driver simulates customizable \textit{Transport-Triggered Architecture (TTA)}~\cite{Corporaal:97:TTA_BOOK}
based accelerators executing the kernels. The processors are simulated by calling
the instruction set simulator of the \textit{TTA-based Co-design Environment (TCE)}~\cite{TCEFPGA-ALT}.
The driver performs the memory management of the device memories at the host side, 
and controls the kernel execution at the device.

\item[\texttt{cellspu}] Another (experimental) heterogeneous accelerator device. This controls a 
single \textit{Synergistic Processing Elements (SPE)} in the heterogeneous Cell~\cite{SPUALT} 
architecture running a Linux-based operating system. It uses the libspe for interfacing with 
the SPE.
 \end{description}

It should be noted that each of the previous device layers provide 
varying levels of portability themselves. For example, the \textit{pthread} device layer
implementation can be used with \textit{Symmetric Multi-Processing (SMP)} systems that run 
an operating system which supports the pthreads API, regardless of the underlying 
CPU architecture. The \textit{ttasim} driver, on the other hand, assumes a specific 
communication mechanism through explicit messages and buffer transfers using DMA commands.


One important responsibility of a device layer implementation is
resource management, that is, ensuring the resources of the device needed for
kernel execution resources are properly shared and synchronized between multiple 
kernel executions. The allocation of the OpenCL buffers from the device memory 
requested via the \textit{clCreateBuffer} and similar APIs is also part of the 
resource management responsibility of the device layer. 


For assisting in memory management, \pocl provides a 
memory allocator implementation called \textit{Bufalloc}
which aims to optimize the allocation of large 
continuous buffers typical in OpenCL applications. 
There are two main motivations for the
customized kernel buffer allocator: 1) exploit the knowledge of the 
``throughput computing'' workloads of OpenCL where the buffers are usually
relatively big to reduce fragmentation, and 2) offer a generic memory allocator for
devices without such support on device.

The working principle of the allocator is similar to \textit{memory pools} in that 
a larger region of memory can be allocated at once with a single \textit{malloc} 
call (or at compile time by allocating a static array). 
Chunks of this region are then returned to the application using a fast allocation 
strategy tailored for the OpenCL buffer allocation requests. 
As the allocation of the initial region can be done in multiple ways, the same 
memory allocator can be also used to manage memory for devices without operating 
systems. In that case, the host only keeps book of all the buffer allocations using 
Bufalloc for a known available region in the device memory and the device assumes 
all the kernel buffer pointers are initialized by the host to valid memory locations.
The memory allocation strategy is designed according to the assumption that 
the buffers are long lived (often for the whole lifetime of the OpenCL application) 
and are allocated and deallocated in groups (space for all the kernel buffer 
arguments reserved and freed with successive calls to the allocator). These
assumptions imply that memory fragmentation can be reduced by 
allocating neighboring areas of the memory for the successive allocation requests.
A simple first fit algorithm is used in finding free space for the buffer allocation
requests. 

The internal book keeping structure of Bufalloc is split to \textit{chunks} with
a free/allocated flag and a size. The chunks are ordered by their starting
address in a linked list. The last chunk in the list is a sentinel that holds
all the unallocated memory. When a buffer allocation request is received,
the linked list is traversed from the beginning to the end until an unallocated
chunk with enough space is found. This chunk is then split to two chunks;
one having the exact size of the buffer request that is returned to the
caller, and another carrying the rest of the unallocated space in the original
chunk. The allocation strategy has a customizable \textit{greedy} mode which 
always serves new requests from the last chunk (end of the region) if possible. 
This mode results more often in the successive kernel buffer allocation calls being 
allocated from continuous memory space given the original allocated region
is large enough. 


\section{Performance Portable Kernel Compiler}
\label{sec:KernelCompiler}

The performance portability in our approach is obtained with an OpenCL kernel compiler 
which exposes the parallelism in the kernels in such
a way that it can be mapped to the diverse parallel resources available in
the different types of computing devices. In this section, we discuss 
the kernel compiler and provide details of the work-group function generation.

\subsection{Compilation Chain Overview}

The pocl kernel compiler is based on unmodified Clang~\cite{clanghome} and LLVM~\cite{llvmhome} 
tools. Clang
parses the OpenCL C kernels and produces an LLVM \textit{Intermediate Representation (IR)}
for the \pocl kernel compiler passes.
The generated LLVM IR 
contains the representation of the kernel code for a single work-item,
matching the original OpenCL C kernel description as an LLVM IR function.

The kernel description can be thought of as a description of a thread which
executes independently by default, and which is synchronized explicitly across
other work-items in the same work-group by the programmer-defined barriers. 
The thread is then spawned as many times as there are work-items in the work-group,
in the \textit{Single Program Multiple Data (SPMD)} parallel program style.

Whether this single work-item program description can be executed directly on the device
depends on the execution model of the target. If the target device is tailored
to the SPMD style of parallelism it might
be able to input a single kernel description and apply the same instructions over multiple data
automatically.
This is the case with many of the GPUs which implement an execution model called 
\textit{Single Instruction Multiple Threads (SIMT\@)}. SIMT devices make it 
the responsibility of the hardware to spread the execution of the kernel description
to multiple work-items that consist the work-group. Each SIMT core contains an independent 
program counter, but share the same instruction feed, so that the same kernel instruction is broadcast
to all the cores with the same program counter value. Thus, the cores wait for their 
own separate part of the kernel in case of diverging execution, and continue with parallel 
execution whenever the work-items converge~\cite{CUDA}. 

For \textit{Multiple Instructions Multiple Data (MIMD)} or architectures with 
\textit{Single Instruction Multiple Data (SIMD)} instructions, on the other hand, the semantics of a multi-WI 
work-group execution must be 
created by the compiler or the threading runtime. A straightforward implementation of
OpenCL kernel execution on a MIMD device would simply spawn as many threads as there
are work-items for the kernel function, and implement the work-group barriers using
barrier synchronization primitives. However, as OpenCL is optimized for high throughput massively 
parallel computation, this type of thread level parallelism is usually too heavy for
work-item execution. Creating, executing, barrier-synchronizing and context switching hundreds or 
thousands of threads for each kernel invocation would incur so large overheads that the performance 
benefits of parallel work-item execution are easily ruined for at least the smaller kernel functions.

Moreover, in order to improve the performance portability of the OpenCL programs, it is desirable
to map the work-items in a work-group over all the parallel resources
available on the device at hand.
For example, if the target supports SIMD instructions as instruction set extensions, the compiler should
attempt to pack multiple work-items in the work-group to the same vector instructions, one work-item
per vector lane. In case of in-order superscalar or \textit{Very Long Instruction Word (VLIW)} style 
\textit{Instruction-Level Parallel (ILP)} architectures it might be beneficial to ``unroll'' the 
parallel regions in the kernel code in 
such a way that the operations of several independent work-items can be statically scheduled 
to the multiple function units of the target device. On the other hand, if vectorization
across the work-group is not feasible, for example, due to excessive diverging control flow in the
kernel, the most efficient way to produce the work-group execution 
might be to execute all the work-items serially using simple loops and rely on 
the work-item vector datatypes for vector hardware utilization. This alternative minimizes the
instruction cache footprint and might still be able to exploit instruction parallel execution of
multiple work-items in case of out-of-order hardware.

An overview of the kernel compilation process of pocl is depicted in Fig.~\ref{fig:kernelCompilationMethods}.
First, the OpenCL kernel (if given in source form) is fed to the Clang OpenCL C frontend 
which produces an LLVM IR of the kernel function for a single work-item. SPIR is an alternative 
input format which allows to skip the Clang phase. 

The LLVM IR function that describes the behavior of a single work-item in the work-group
is then processed by the pocl's kernel compiler, which links the IR against
an LLVM IR library of device-specific OpenCL built-in function implementations at
the bitcode level. The function is converted to a \textit{work-group function} 
(in case of a non-SPMD execution model target) that generates a
version of the function that statically executes all the work-items of the 
work-group. This is done using the work-group function generation passes of pocl.

When compiling to SPMD-optimized hardware such as SIMT GPUs\footnote{
At the time of this writing, pocl does not yet support popular commercial GPU targets.
However, the SPMD/GPU path of the kernel compiler has been tested by using
research targets to ensure GPU-like devices can be supported using pocl.}, the generation 
of the work-group function is not necessary as the hardware produces the 
multiple parallel work-item execution. However, it is sometimes still 
beneficial to merge multiple work-items to expose instruction-level 
parallelism in case the cores contain multiple function units. 

The work-group function is a version of the original kernel with data parallel 
regions across the independent work-items exposed to the later phases of the compiler. 
It consists of parallel ``work-item loops'' that execute so called ``parallel regions'' 
for all work-items. The parallel regions are formed according to the barriers in the 
kernel.

Currently the work-group function generation is performed at \textit{kernel enqueue}
time, when the local size is known. The known local size makes it possible to set 
constant trip counts to the work-item loops, leading to easier static parallelization
later on in the compilation chain. For example, a vectorizer can then easily see whether 
the trip counts can be covered evenly with the maximum size vector instructions in 
the machine. Otherwise, it would always need to create a copy of the loop that iterates 
the ``overhead iterations'' that could not be covered evenly with the vector instructions. 
The drawback of this approach is that one work-group function needs to be generated 
for each local size. If the same kernel is executed with a lot of different local sizes,
it leads to compilation time increase. We have not seen this is as a problem yet in 
our test cases. However, it would be trivial to add a version of the work-group function with 
variable trip counts in the work-item loops to produce a work-group function that can 
be used with all local sizes, but might not be so efficiently parallelizable.

The produced work-group function is in a format that can be launched for
different parts of the work-space in parallel. This can be seen as an additional 
struct function argument added to the work-group function that contains the work-space
coordinates among other information. In addition, the automatic local array visible
in the original kernel source is converted to a function argument to unify the handling of 
local buffers which can be allocated both by the host and by the kernel.

Finally, the work-group function is passed to the code generator and assembler which
generate the executable kernel binary for the target device. The work-group 
function is potentially accompanied with a launcher function in case of a heterogeneous 
device. In that case the device contains its own main function which executes the
work-group function on-demand.

\begin{figure}[tp]
\begin{center}
\includegraphics[width=0.87\textwidth]{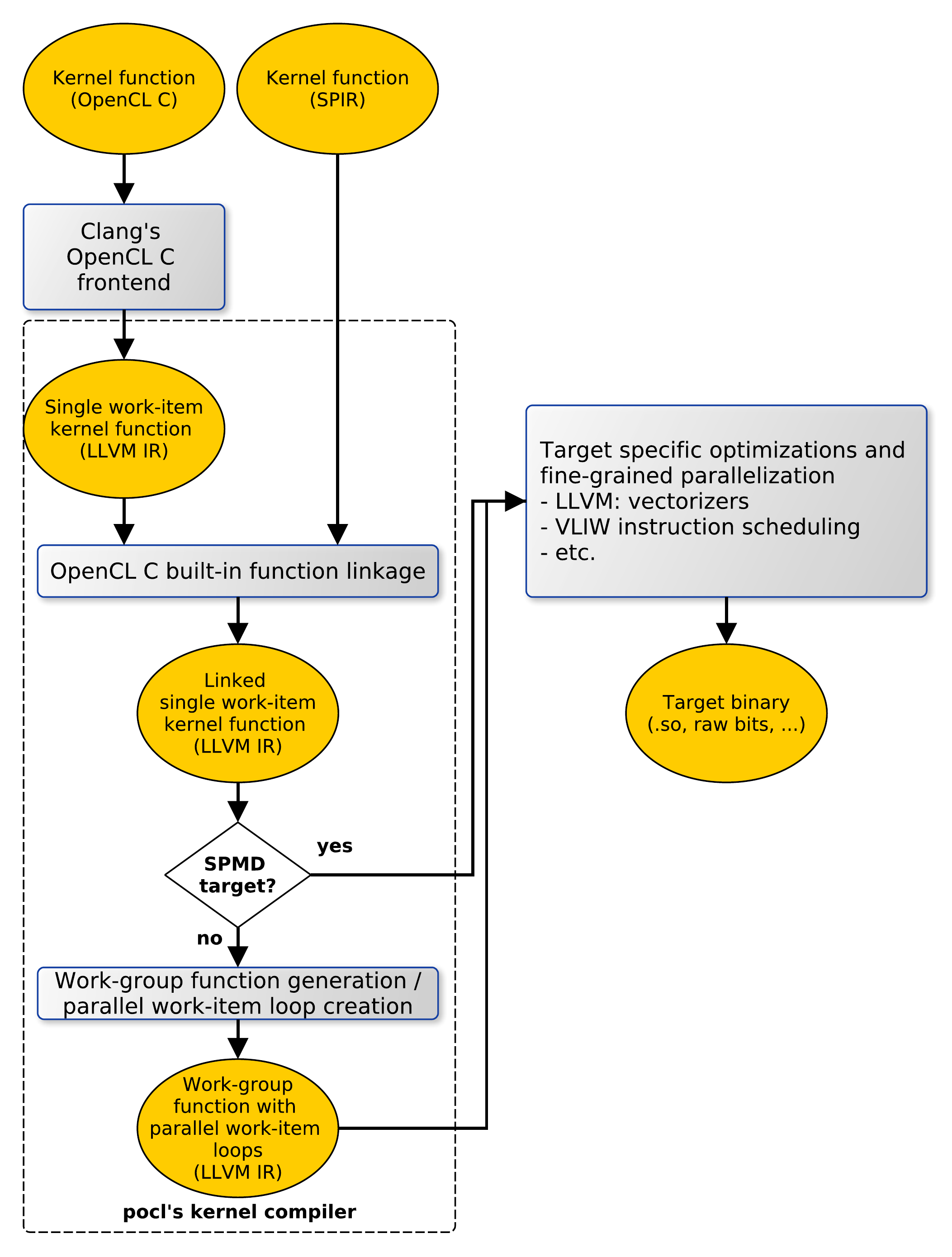}
\caption{A high level illustration of pocl's kernel compilation chain. The source
code of the kernel is read by Clang which produces an LLVM IR for the single work-item kernel 
description. Alternatively, a pre-built SPIR bitcode binary can be used as an input. The
OpenCL C built-in functions are linked at the LLVM IR level to the kernel after 
which the optional work-group function generation is done. In case the target can 
execute the SPMD single work-item kernel description directly for all work-items in 
the work-group (as is the case with most GPUs), or
the local size is one, this step is skipped. The work-group function generation 
is the last responsibility of the pocl's kernel compiler; it helps the later 
target-specific passes (such as vectorization) by creating parallel work-item 
loops which are annotated using LLVM metadata.
}
\label{fig:kernelCompilationMethods}
\end{center}
\end{figure}

\subsection{Generation of Parallel Work-Group Functions}

The main responsibility of the kernel compiler of pocl is generating a new
work-group function out of the single work-item kernel 
function that is produced by the Clang frontend. The work-group function
executes the kernel code for all the work-items in a work-group of 
a given size and exposes the parallel parts between work-items in a way
that can be potentially exploited by a target specific vectorization pass or 
an instruction scheduler/bundler. In practice, the parallel loops are annotated 
with LLVM metadata that retains the information of the parallel iterations for 
later phases such as the loop vectorizer, which then does not
have to prove the independence of the loop iterations to perform vectorization. 

Producing multi-WI work-group functions is not trivial due to the need
to respect the synchronization semantics of the work-group barriers inside 
the kernel code. That is, the multiple work-item execution cannot be 
implemented by simply adding a loop around the kernel code that executes
the function for all the work-items, but the regions between the
barriers must be parallelized separately.
Statically parallelizing kernels with barriers inside conditional regions 
such as for-loops or if-else-structures adds further complexity. In such cases 
the regions between barriers are harder to parallelize due to the varying paths the
execution can take to the barrier call. 
The pocl kernel compiler is modularized to parts that are reusable across 
several parallelization methods.  For example, 
all of the methods for implementing static parallel computation on the device 
need to identify the \textit{parallel regions} in the kernels (regions 
between barriers) which can be then mapped to the parallel resources in 
multiple ways. 

Throughout the following algorithm descriptions the kernels are 
represented as \textit{Single Static Assignment (SSA)}~\cite{computingSSA} 
\textit{Control Flow Graphs (CFG)}~\cite{CFG} of the LLVM IR\@. 
The relevant characteristics of the internal representation are
as follows:

\begin{itemize}
\item  Variable assignments and operations are abstracted 
  as \emph{instructions}. Instructions
  have at most one result, and referring to an
  instruction means referring also to its output value if
  it exists. 

\item A node in a CFG is a \emph{Basic Block (BB\@)}.
  A BB is a branchless sequence of \emph{instructions} 
  which is always executed as an entity, from the
  first instruction to the last.
  
\item An edge in a CFG represents a branch in the control
  flow. These edges are defined by the jump instructions in
  the \emph{source} BB\@. This implies that creating a copy $B'$ of a
  basic block $B$ which has an edge to basic block $C$
  results in $B'$ also having an edge to $C$.
  
\item Both the source and the destination BBs of
  any CFG edge belong to the CFG\@. This is important characteristics
  in order to differentiate between a CFG and a sub-CFGs, 
  defined below.
  
\item Multiple \textit{exit BBs} are allowed. Typically the
 exit BBs are blocks that return from the function at hand.

\end{itemize}

We also define the term \emph{sub-CFG}, to refer to a
CFG which is a subgraph of another CFG\@. A
sub-CFG always has an associated CFG, and has essentially
the same properties as CFGs, save that it might have edges leading
to blocks that do not belong to the sub-CFG but to the
parent CFG\@.

There are two helper functions which are 
used in the algorithm descriptions later in this section.
Function \textit{CreateSubgraph} finds all the nodes
which can potentially be visited when traversing from node $A$
to node $B$. This function can be used to construct a single-entry
single-exit subgraph between two given nodes. It can be implemented 
with a depth-first 
search starting from the desired subgraph entry $A$ and keeping record 
of all the nodes visited when traversing all the possible paths to
the subgraph exit node $B$ and by ignoring edges back to an
already visited node to avoid infinite loops.

%

Function \textit{ReplicateCFG} takes a CFG or a sub-CFG and
replicates the whole graph. Thus, it copies both the BBs in
the and their edges creating an identical copy of the graph
as a whole.

\subsection{Parallel Region Formation}

The generation of static multi-WI work-group
functions involves identifying the regions between barriers
that must be executed by all the work-items before proceeding
to the next region. These regions are referred to as
\textit{parallel regions} or simply \textit{regions} in
the rest of this article.  

The work-items in the work-group can execute the code in the parallel regions 
in any order relative to each other due to the the relaxed consistency model of 
the device memory in the OpenCL 1.2 standard~\cite{OpenCL12}. Thus,
the multi-WI functions can be implemented as ``embarrassingly
parallel'' loops that iterate over all the work-item local ids
with a parallel region as the loop body. The parallel loops (later referred
to as \textit{WI loops} as in work-item loops)
often form the main source of fine-grained parallelism 
available for the parallel computation resources in the targeted device.

The simplest scenario for forming the parallel regions is a 
kernel without barriers. In such a case, creating a work-item loop whose 
body is the CFG of the whole function is sufficient (see 
Figure~\ref{fig:simpleloop}(a)). Thus, a single parallel region
consisting the whole kernel function is formed. The only requirement for 
this to work is that the original CFG has a single entry point; 
(this is always true for the kernel functions as the function can be
entered only from one location), and a single exit point. The latter can 
be achieved by a normalization transformation on the kernel function. 

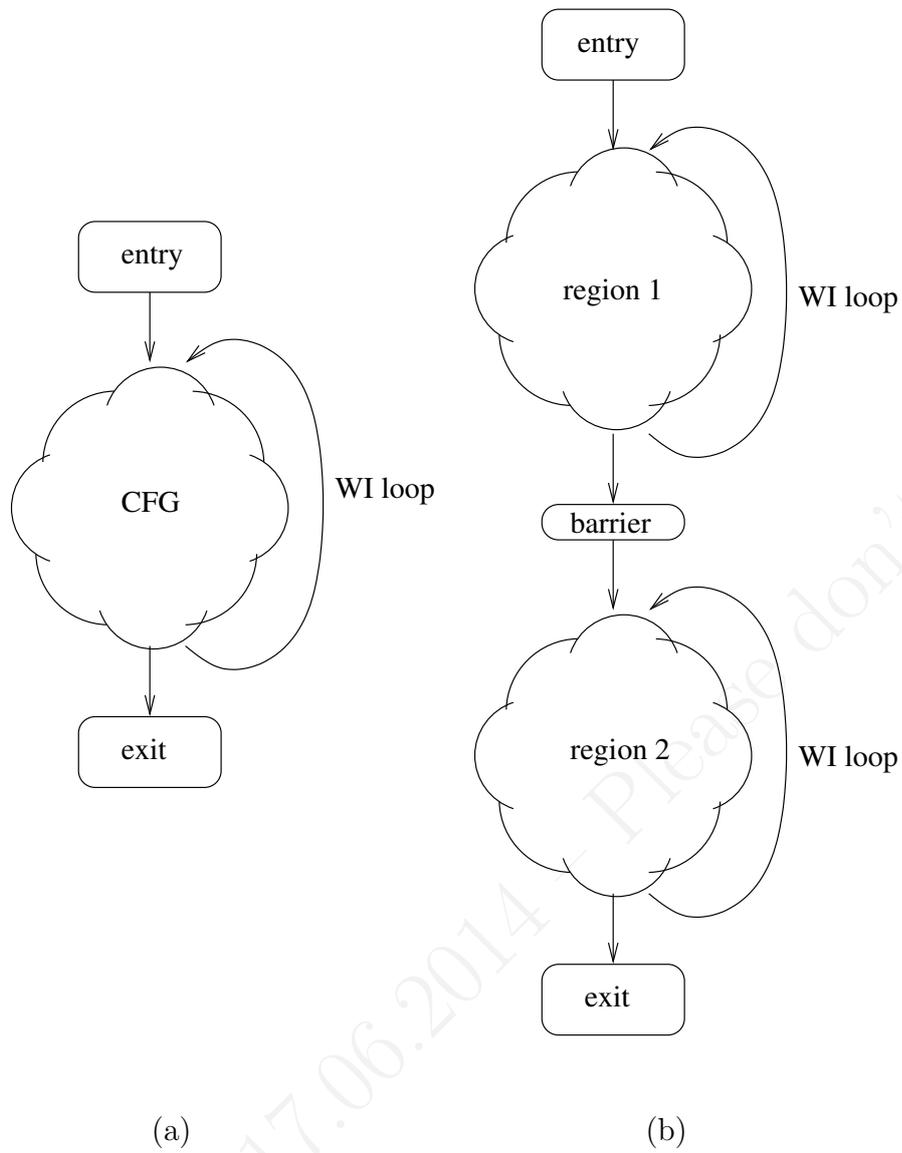
\begin{figure}
\begin{center}

\begin{picture}(0,0)%
\includegraphics[width=1.0\linewidth]{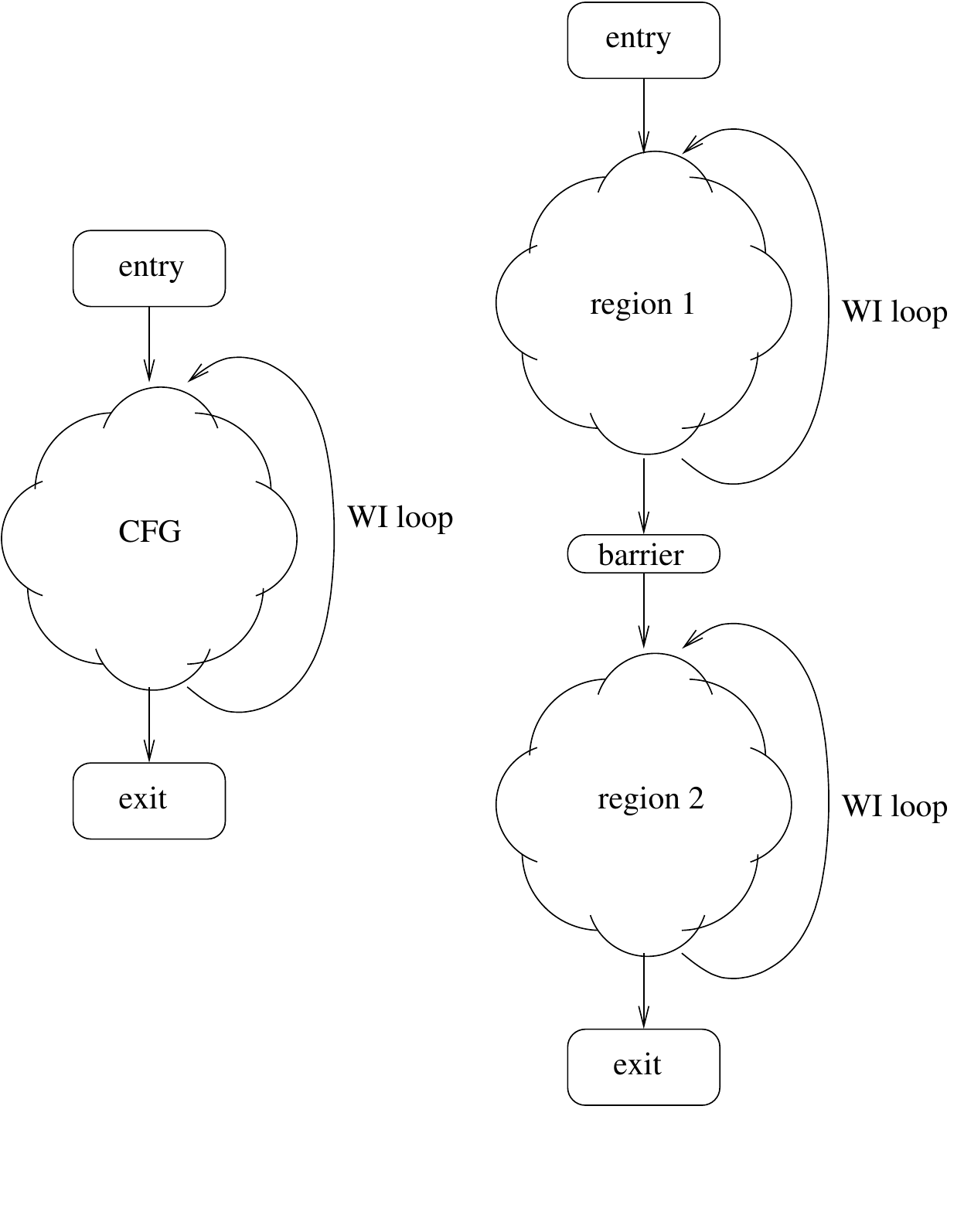}%
\end{picture}%
\setlength{\unitlength}{4144sp}%
\begingroup\makeatletter\ifx\SetFigFont\undefined%
\gdef\SetFigFont#1#2#3#4#5{%
  \reset@font\fontsize{#1}{#2pt}%
  \fontfamily{#3}\fontseries{#4}\fontshape{#5}%
  \selectfont}%
\fi\endgroup%
\begin{picture}(5685,7285)(18,-6884)
\put(3781,-6811){\makebox(0,0)[lb]{\smash{{\SetFigFont{12}{14.4}{\familydefault}{\mddefault}{\updefault}{\color[rgb]{0,0,0}(b)}%
}}}}
\put(856,-6811){\makebox(0,0)[lb]{\smash{{\SetFigFont{12}{14.4}{\familydefault}{\mddefault}{\updefault}{\color[rgb]{0,0,0}(a)}%
}}}}
\end{picture}

\caption{Two basic cases of static work-group function generation: A kernel \newline
(a) without work-group barriers and (b) with an unconditional barrier in 
the middle.}
\label{fig:simpleloop}
\end{center}
\end{figure}


The parallel region formation for kernels with barriers is more complex. 
In the following, 
the work-group barriers are classified to
two categories. If a barrier is reached in all the execution paths of the 
kernel control flow, that is, if the barrier \emph{dominates} the exit node, 
we call it an \emph{unconditional barrier}. In case the barrier is placed
inside a conditional BB such as an if\ldots else structure or a for-loop
(the barrier does \emph{not dominate} the exit node), we call it
a \emph{conditional barrier}. 
Unconditional barriers create separate
\emph{parallel regions}, sections of the
CFG which the different work-items can execute
in parallel. In Figure~\ref{fig:simpleloop}(b),
the unconditional barrier divides the whole CFG
into two regions. In order to comply with the barrier semantics,
no work-item should execute the \emph{region 2} until
all of them have finished executing the \emph{region 1}.
Thus, two WI loops must be created, one iterating 
over each parallel region; one before, and one after 
the barrier.
The parallel region formation algorithm for kernels with
only unconditional barriers is given in Algorithm~\ref{alg:uncondBarrierRegionFormation}.

\begin{algorithm}
\begin{enumerate}

\item Ensure there is an implicit barrier at the entry and the exit nodes
of the kernel function and that there is only one exit node in the
kernel function. This is a safe starting condition as it does not affect
any execution order restrictions.

\item Perform a depth-first-search traversal of the kernel CFG\@.
Ignore the possible back edges to avoid infinite loops and to
include the loops of the kernel to the parallel region.

\item When encountering a barrier, create a parallel region by
calling \textit{CreateSubgraph} for the previously encountered barrier
and the newly found barrier.

\end{enumerate}
\caption{Parallel region formation when the kernel does not contain conditional barriers.}
\label{alg:uncondBarrierRegionFormation}
\end{algorithm}

\subsection{Handling of Conditional Barriers}

The algorithm for parallel region formation described so far can only
handle kernels which either have no barrier synchronization at all, 
or have only unconditional barriers. 
According to OpenCL specification, "the work-group barrier must be encountered by all work-items
of a work-group executing the kernel or by none at all"~\cite{OpenCL12}.
In order to describe the way how \pocl handles kernels with conditional
barriers, a few new definitions are needed.

\begin{figure}
\begin{center}
\begin{picture}(0,0)%
\includegraphics[width=1.0\linewidth]{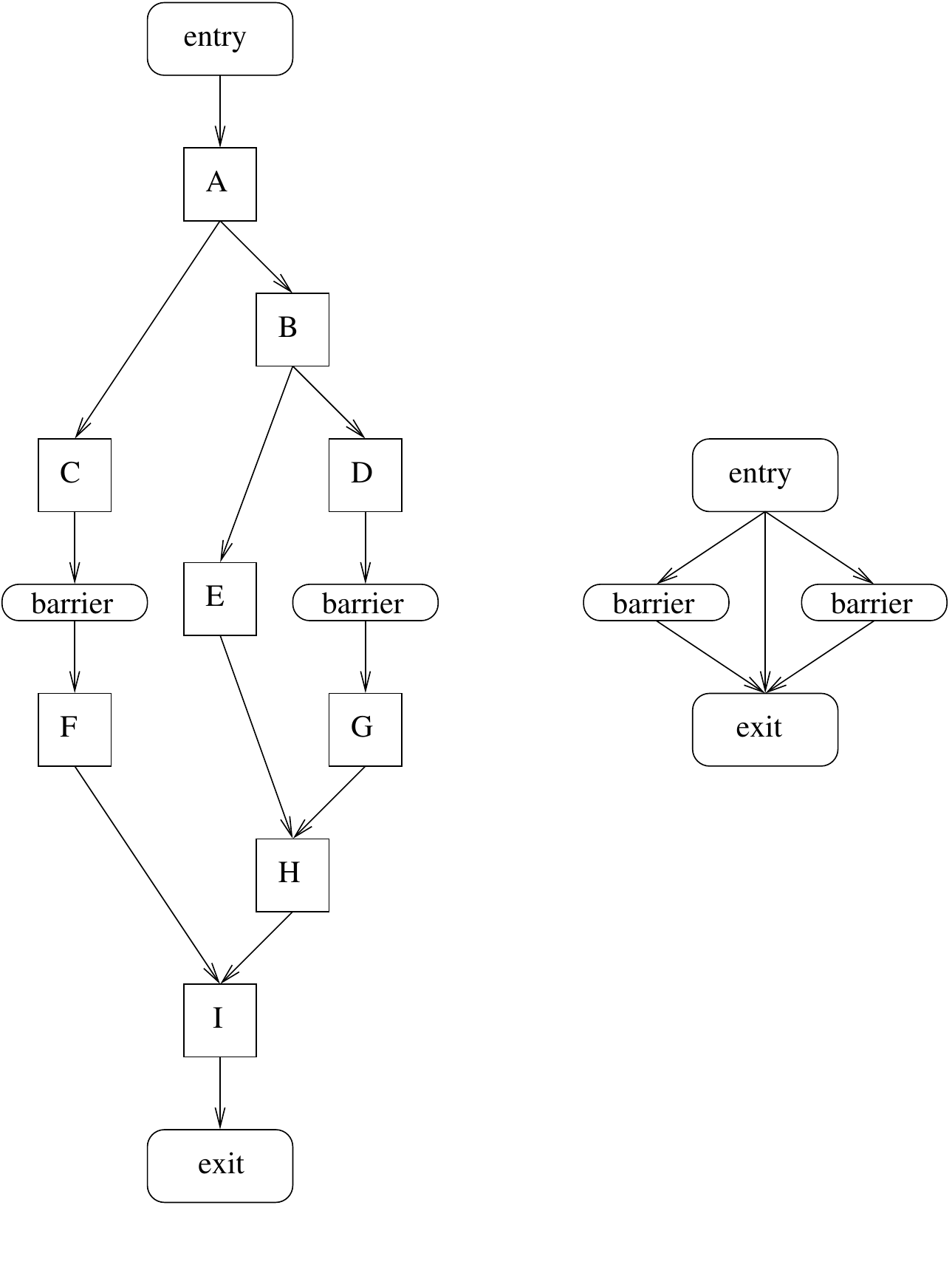}%
\end{picture}%
\setlength{\unitlength}{4144sp}%
\begingroup\makeatletter\ifx\SetFigFont\undefined%
\gdef\SetFigFont#1#2#3#4#5{%
  \reset@font\fontsize{#1}{#2pt}%
  \fontfamily{#3}\fontseries{#4}\fontshape{#5}%
  \selectfont}%
\fi\endgroup%
\begin{picture}(5874,7870)(889,-7019)
\put(5536,-6946){\makebox(0,0)[lb]{\smash{{\SetFigFont{12}{14.4}{\familydefault}{\mddefault}{\updefault}{\color[rgb]{0,0,0}(b)}%
}}}}
\put(2161,-6946){\makebox(0,0)[lb]{\smash{{\SetFigFont{12}{14.4}{\familydefault}{\mddefault}{\updefault}{\color[rgb]{0,0,0}(a)}%
}}}}
\end{picture}%

\caption{(a) An example CFG with two conditional barriers and
  (b) its reduced barrier CFG\@.}
\label{fig:cfg}
\end{center}
\end{figure}
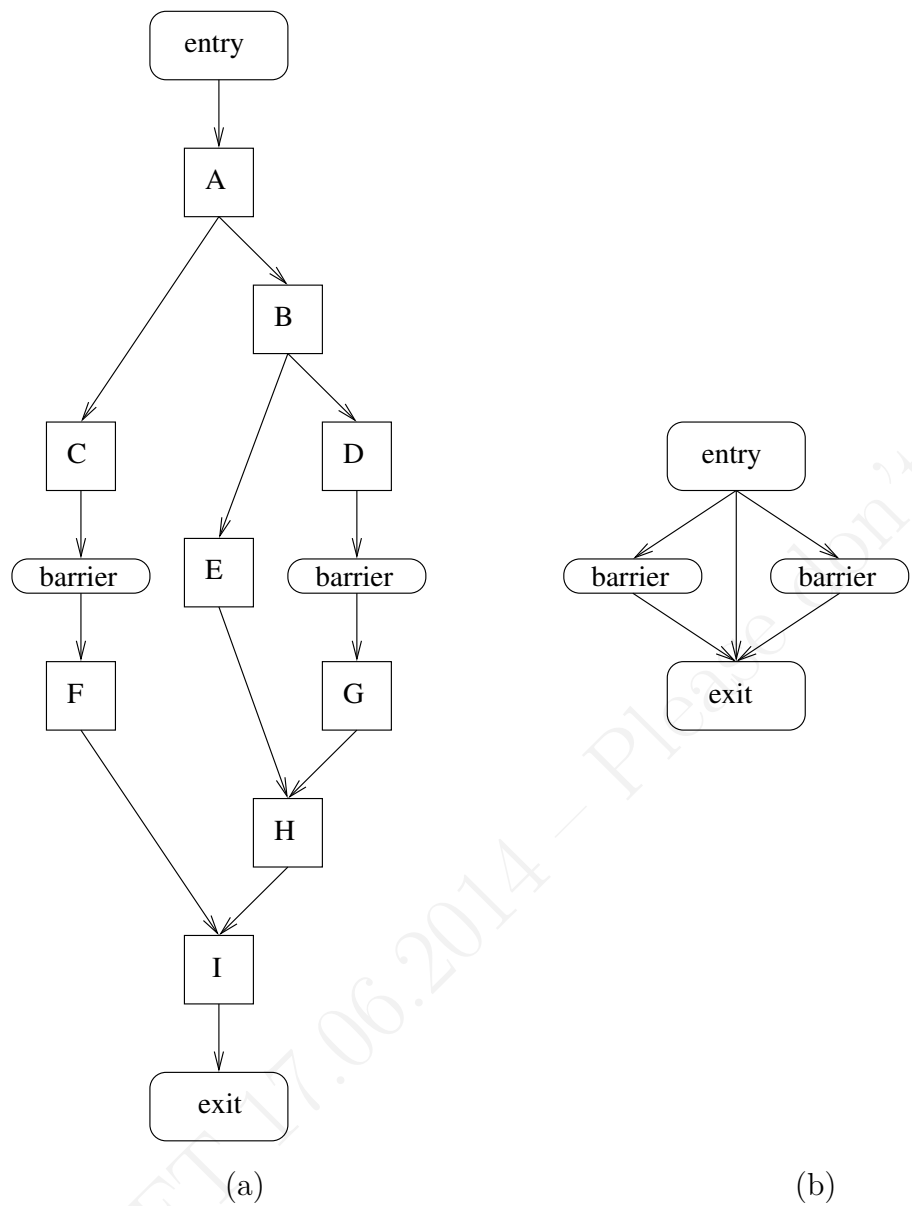


%

\begin{definition}[Barrier CFG] A reduced CFG with all the
non-barrier instructions and basic blocks eliminated. An example
barrier CFG is shown in Figure~\ref{fig:cfg}(b). 
The barrier CFG
is formed by producing a graph with only the barrier, exit and
entry nodes of the original CFG\@. There is an edge between two nodes if and 
only if there is a direct (no-barrier) path between the two nodes in the 
original CFG. Exit and entry nodes contain implicit barriers.

\end{definition}

\begin{definition}[Predecessor barrier] Given a barrier $b$, its
predecessor barriers are all the barriers which can be visited 
in a path leading to $b$ from the entry node.
They correspond to predecessor nodes in the reduced Barrier CFG\@.
Every barrier except the implicit barrier at the entry node
has at least one predecessor barrier.
\end{definition}

\begin{definition}[Successor barrier] Given a barrier $b$, its
successor barriers are all the barriers that might be reached on
any path from $b$ to the exit node.
They correspond to successor nodes in the reduced Barrier CFG\@.
Every barrier except the implicit barrier at the exit node 
has at least one successor barrier.
\end{definition}

\begin{definition}[Immediate predecessor barrier] A barrier node
preceding a given barrier node in the Barrier CFG\@.
\end{definition}

\begin{definition}[Immediate successor barrier] A barrier node
succeeding a given barrier node in the Barrier CFG\@.
\end{definition}

Assuming there is at least one exit node in the kernel function, we can 
state that:

\begin{proposition}
If there is a conditional barrier in a kernel CFG, then there is
at least one other barrier which has more than one immediate
predecessor barrier.

\label{pro:condbarriers}
\end{proposition}

\begin{proof}

Let $U=\{u_i\}$ be the set of all the unconditional barriers in the
Barrier CFG and $C=\{c_j\}$ the non-empty set of all the conditional barriers.
The implicit barrier on the exit node has to be in $U$ (exit node
is not conditional), thus there is at least one edge $e$ from
a node $c_j \in C$ to a node $u_i \in U$, otherwise
there would be no path from any conditional barrier to the
exit node.  This would make all the nodes in $C$ dead or
unreachable basic blocks because if a node in $C$ is executed
at least by one work-item, all work-items must execute it
after which the control shall proceed. Otherwise, as there is
only one exit node in the kernel CFG, there must be an infinite loop after 
the conditional barrier and the kernel outcome is undefined.
Moreover, $e$ can not be the only edge leading to $u_i$, as then
$c_j$ would dominate $u_i$, which could only happen if both $c_j$ and $u_i$ were
of the same kind (conditional or unconditional).
Hence, there are at least two different edges leading
to $u_i$ in the Barrier CFG, thus $u_i$ has at least two
immediate predecessor barriers.  
\end{proof}

%
%


A barrier with two predecessor barriers makes it impossible
to apply Algorithm~\ref{alg:uncondBarrierRegionFormation} for 
forming the parallel regions.
According to the simple algorithm, upon reaching a conditional barrier, 
a parallel region should be formed between the preceding barrier and 
the previously reached one, but in this case there would be ambiguity 
on which one is the preceding barrier. For example, in Figure~\ref{fig:cfg}(a),
when reaching the exit node, the work-item loop iterating over
the parallel region might have to
branch back to either \emph{A}, \emph{F}, or \emph{G}, depending
on the execution path chosen by the first work-item (which the
other work-items must follow, according to the OpenCL work-group
barrier semantics).

In order to form a single entry, single exit parallel
regions in the presence of conditional barriers, we
apply a variant of \textit{tail duplication}~\cite{TailDuplication} to the set of 
basic blocks reachable from the conditional barrier at hand.
This produces a new CFG with the same behavior
as the original CFG, but in which each barrier can have
only one immediate predecessor barrier, enabling the
single entry single exit parallel region formation
similarly as with unconditional barriers.
The used tail duplication process is described in 
Algorithm~\ref{alg:tailreplication}.

\begin{algorithm}[h!]

\begin{enumerate}
\item Perform a depth-first traversal of the CFG, starting at the entry node.
\item Each time a new, unprocessed conditional barrier is found, use \textit{CreateSubgraph}
to produce a sub-CFG from that barrier to the next exit node (duplicate
the tail).
\item Replicate the created sub-CFG using \textit{ReplicateCFG}.

In order to
reduce code duplication, merge the tails from the same unconditional barrier 
paths. That is, replicate the basic blocks only after the last barrier that 
is unconditionally reachable from the one at hand.
\item Start the algorithm again at each of the found barrier successors.
\end{enumerate}

\caption{Tail duplication for parallel region formation in the case of
conditional barriers in the kernel.}
\label{alg:tailreplication}
\end{algorithm}

The result of applying tail replication to the example CFG
in Figure~\ref{fig:cfg} is shown in Figure~\ref{fig:aftertails}(a). From
its reduced barrier CFG (Figure~\ref{fig:aftertails}(b)) it can be seen
that no barrier has more than one immediate predecessor barrier after
this transformation has been performed, thus making the parallel
region formation unambiguous.

\begin{figure}
\begin{center}
\begin{picture}(0,0)%
\includegraphics[width=1.2\linewidth]{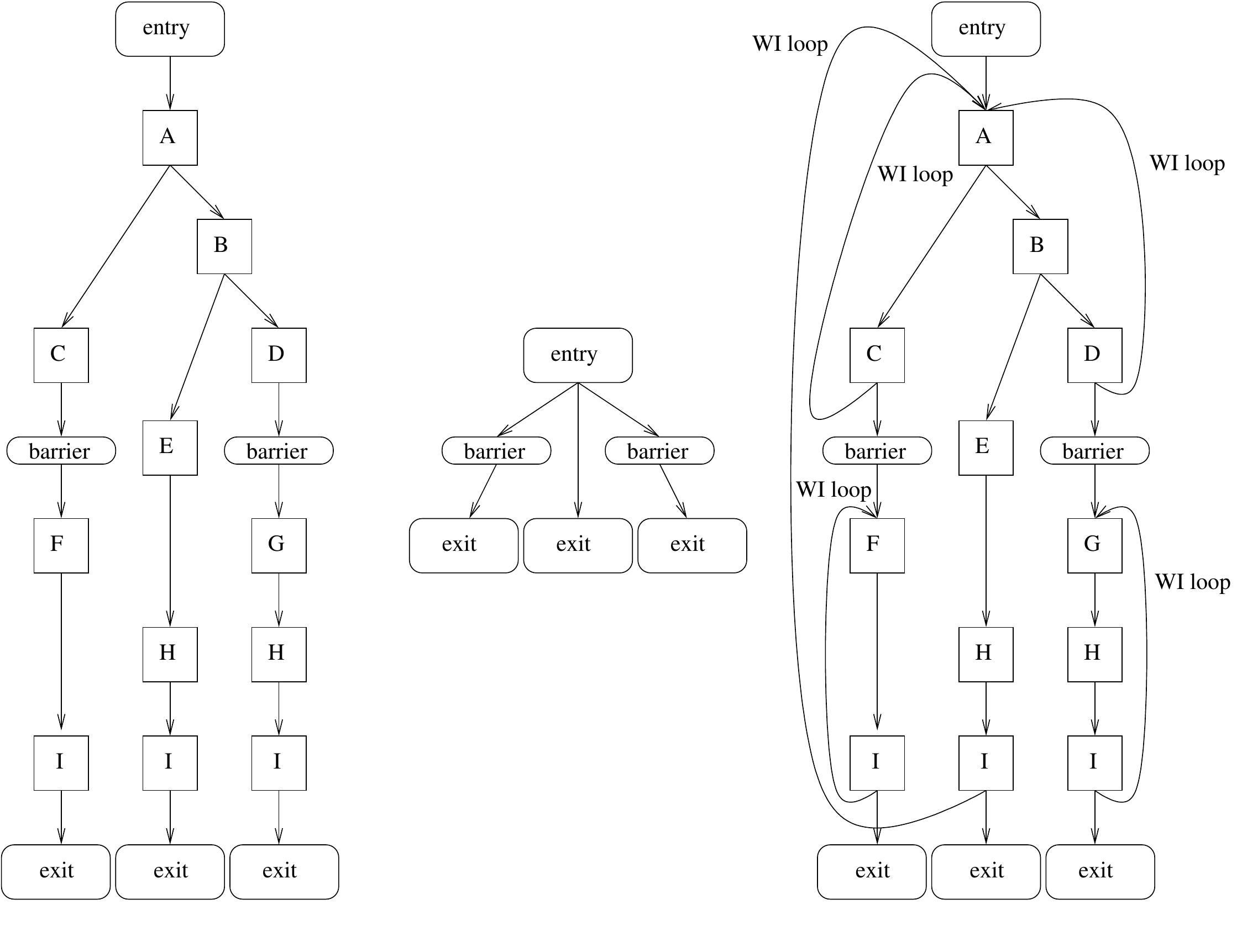}%
\end{picture}%
\setlength{\unitlength}{4144sp}%
\begingroup\makeatletter\ifx\SetFigFont\undefined%
\gdef\SetFigFont#1#2#3#4#5{%
  \reset@font\fontsize{#1}{#2pt}%
  \fontfamily{#3}\fontseries{#4}\fontshape{#5}%
  \selectfont}%
\fi\endgroup%
\begin{picture}(10259,7870)(844,-7019)
\put(6000,-6946){\makebox(0,0)[lb]{\smash{{\SetFigFont{12}{14.4}{\familydefault}{\mddefault}{\updefault}{\color[rgb]{0,0,0}(c)}%
}}}}
\put(3600,-6946){\makebox(0,0)[lb]{\smash{{\SetFigFont{12}{14.4}{\familydefault}{\mddefault}{\updefault}{\color[rgb]{0,0,0}(b)}%
}}}}
\put(1800,-6946){\makebox(0,0)[lb]{\smash{{\SetFigFont{12}{14.4}{\familydefault}{\mddefault}{\updefault}{\color[rgb]{0,0,0}(a)}%
}}}}
\end{picture}%

\caption{(a) Example CFG after tail replication,
  (b) its reduced barrier CFG, and
  (c) parallelized version.}
\label{fig:aftertails}
\end{center}
\end{figure}
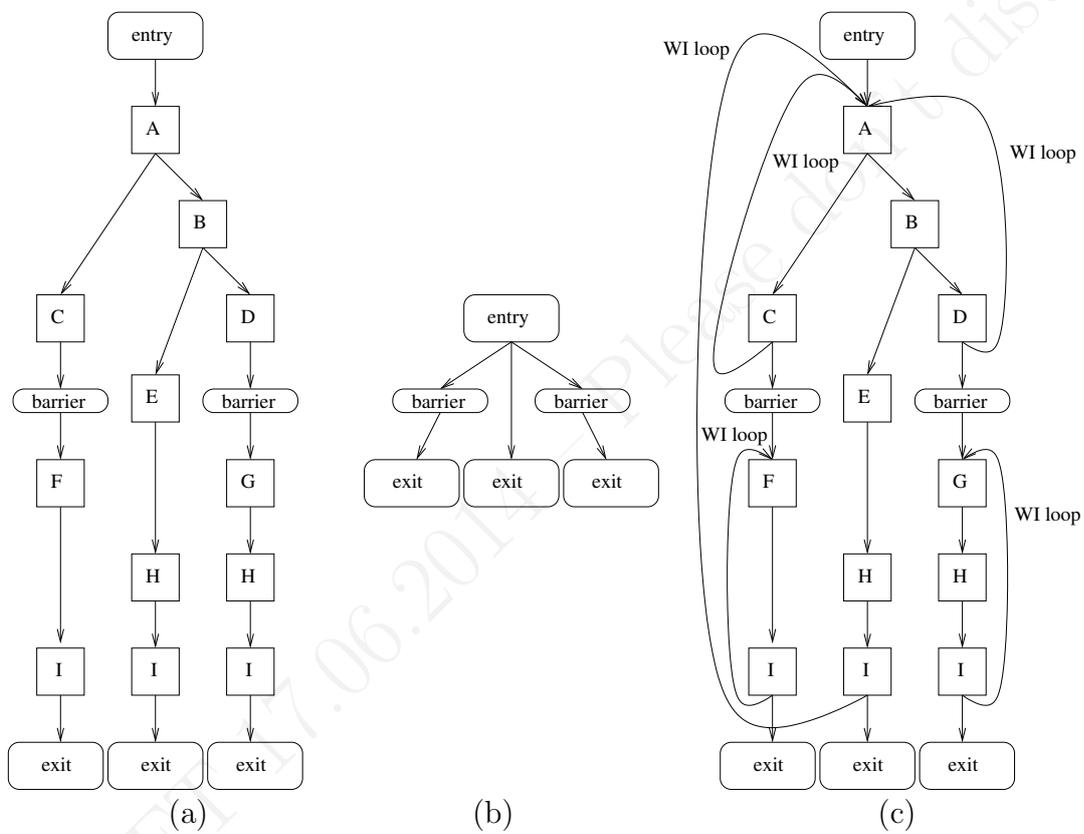

It should be noted that the resulting tail replicated graph has irreducible 
loops~\cite{CFGReducibility}; multiple work-item loops share the same 
basic blocks which leads to branches from a work-item loop to another. 
For example, the basic blocks
$A$, $B$ and $D$ form a parallel region and from $B$, there's a branch to 
the middle of another parallel region's ($ABEHI$) work-item loop. 
Removing branches from a work-item loop to another can be done by leaning on the 
definition of the OpenCL C work-group barriers: if at least one work-item
takes the branch after $B$ that can lead to a barrier, the rest of the work-items 
must follow. This fact can be exploited by ``loop peeling'' the first
iteration of the work-item loop. This iteration is then the only one that
evaluates the work-item dependent condition that chooses which parallel region
should be executed by the rest of the work-items. 
Figure~\ref{fig:afterpeeling} depicts the CFG after loop peeling has
been applied to the conditional barrier parallel regions. The peeled basic
blocks are marked with dashed outline boxes. The peeled paths select the
parallel region work-item loop that is then executed with the branch
selecting the conditional barrier removed. The benefit for parallelization
is apparent; for static multi-issue ILP targets the work-item loops
contain now longer branchless traces from which to issue instructions to
the parallel function units. In general, longer branchless traces
produce more freedom to the compiler instruction scheduler which helps
to hide latencies. 

\begin{figure}
\begin{center}
\includegraphics[width=0.3\textwidth]{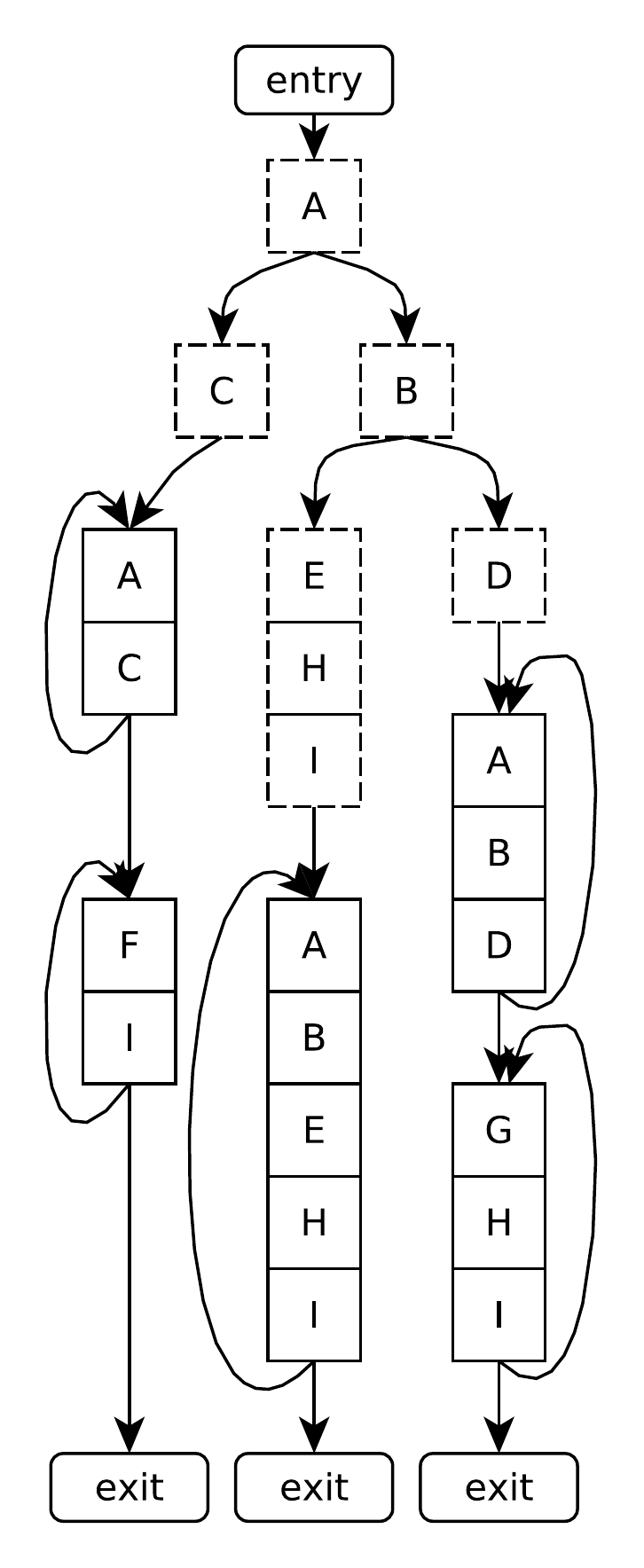}
\caption{The kernel CFG after loop peeling applied to remove irreducible control 
flow from the work-item loops. The ``peeled'' basic blocks are marked with dashed
boxes. This CFG does not contain the explicit barrier markers
as the work-item loops itself implement the work-group barrier semantics.}

\label{fig:afterpeeling}
\end{center}
\end{figure}

\subsection{Barriers in Kernel Loops}

OpenCL allows kernel loops to have barrier synchronization inside loops. 
The semantics of a loop with a barrier (later referred to as \textit{b-loops}) is 
similar to the conditional barriers: if one work-item reaches the barrier, 
the rest of them have to. The barrier call at each kernel loop iteration is 
considered to be a separate barrier instance, that is, the barrier of each iteration 
must be reached by all the work-items before proceeding to the next iteration.
The parallel region formation for b-loops can be reduced to the ``regular''
parallel region formation case by adding certain implicit barriers to the
loop construct. The implicit barriers are added using the following
assumptions:

\begin{enumerate}
 \item All OpenCL kernel loops can be converted to natural canonical loops which have 
a single entry node, the loop \textit{header}, that computes the loop condition 
and just one loop \textit{latch} which jump back to the loop header.
This can be assumed because the OpenCL standard declares kernels with irreducible 
control flow implementation-defined~\cite{OpenCL12} and it is possible to
convert irreducible loops (e.g. those produced by an earlier optimization)
to reducible loops, e.g., via node 
splitting~\cite{ControlledNodeSplitting}. Additional transformations (included
in LLVM passes) can canonicalize loops, ensuring that they have exactly one
back edge.

 \item All work-items execute the iterations of b-loops in lock-step, one
parallel region at a time. Thus, the loop iteration count is the same for all
work-items executing the b-loop. 

 \item If the b-loop has early exits, they have been converted to converge
to a single loop exit basic block.

\end{enumerate}

With the above assumptions, the following implicit barriers can be 
added in order to ensure unambiguous parallel region formation 
for b-loops:

\begin{enumerate}

 \item End of the loop pre-header block. This is the single block preceding the
loop header. That is, synchronize the work-items just before entering the b-loop.

 \item Before the loop latch branch. The original loop latch branch is
retained, thus a parallel region must be formed before it and the original loop
branch preserved.

 \item After the \textit{PhiNode} region of the loop header block. This creates a parallel
region for updating the induction variables and other loop-carried variables in the
original kernel. 

Due to the b-loop iteration-level lock step semantics, the induction 
variable updates are redundant for all the work-items and can be combined by 
the standard \textit{common subexpression elimination}~\cite{GCSE} optimization implemented
by the LLVM\@. Depending on the target, however, the induction variables of the work-items 
might not be beneficial to be combined to a single variable, but duplicated, to avoid the need
to broadcast the single induction variable across all the vector lanes.
\end{enumerate}

\begin{figure}
\begin{center}
\includegraphics[width=0.6\textwidth]{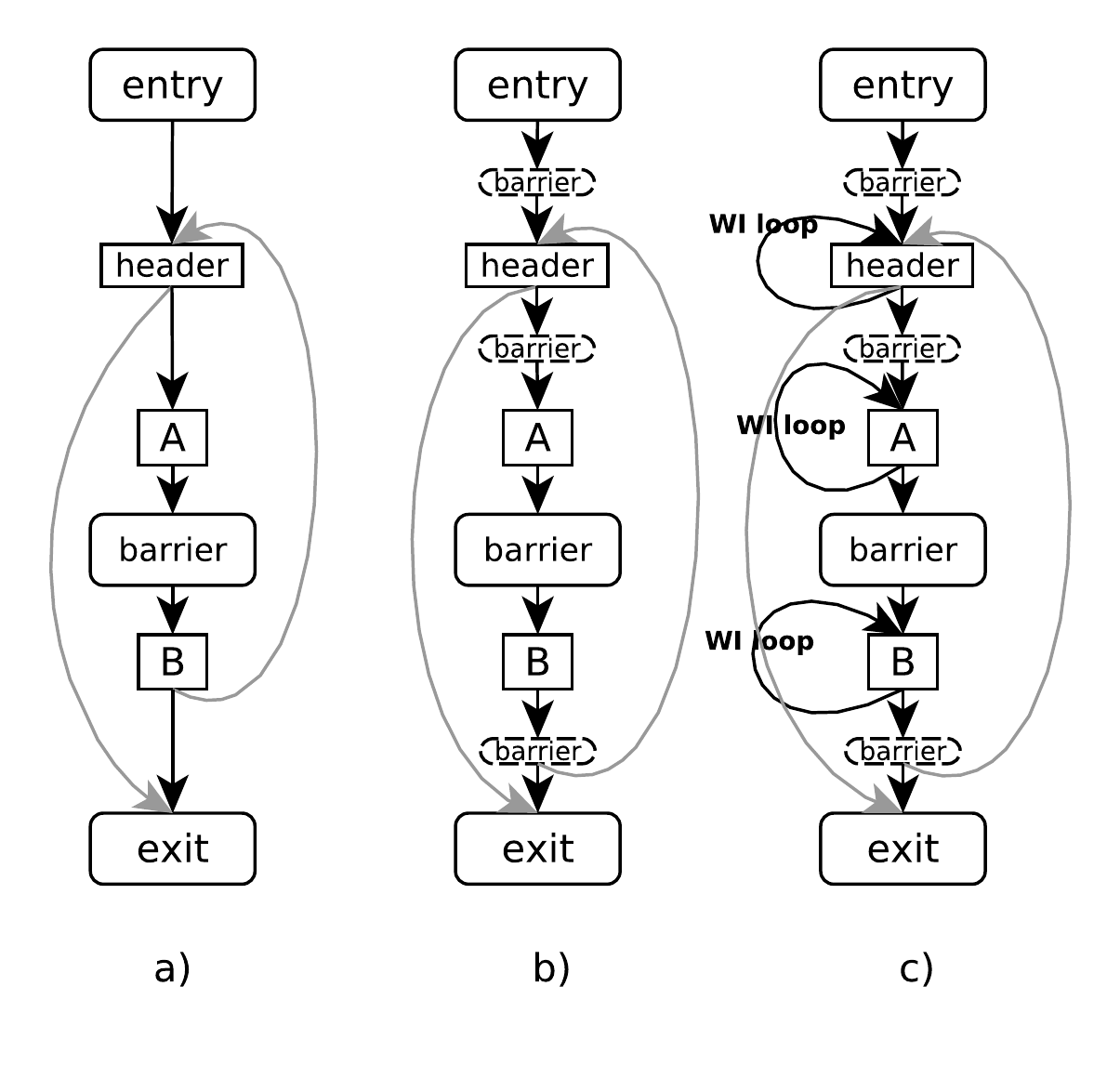}
\caption{Adding implicit barriers to kernels with b-loops to produce
unambiguous parallel regions;
a) the original single work-item kernel CFG with the b-loop,
b) the kernel CFG with implicit barriers added to make parallel 
region formation unambiguous, and
c) the work-group function CFG with the work-item loops added to 
iterate the parallel regions.
The original kernel loop edges are colored grey.
}
\label{fig:implicitLoopBarriers}
\end{center}
\end{figure}

Figure~\ref{fig:implicitLoopBarriers} shows how the implicit barriers
direct the parallel region formation in a kernel with a b-loop. The explicit
(programmer-defined) barrier is shown with a solid outline box, and the
implicit barriers added by the compiler are highlighted with
dashed boxes. It should be emphasized that the original b-loop branches 
in the single work-item kernel (the gray edges in Fig.~\ref{fig:implicitLoopBarriers}) 
are not replicated during the work-group function generation. This enforces 
the semantics of the iteration level lock step execution of b-loops: When a single
work-item stops iterating the loop or begin a new iteration, so shall the others.

\subsection{Horizontal Inner-loop Parallelization}
\label{sec:horizontalLoopParallelization}

The loop constructs in OpenCL C kernel descriptions, written
by the programmer, are like C loops with sequential execution semantics. 
Therefore, in order to parallelize the loops the same loop carried dependency
analysis as in sequential programs is needed. 
In case of multi-WI work-groups, these ``inner loops'' can be
sometimes parallelized ``horizontally'' across work-items in the work
group, thus leading to a more easily parallelized program (the
work-item loop is a parallel loop). In other words, the loop iterations
could be executed in lock step for each work-item before progressing
to the next iteration. 
For example, the imaginary kernel in Fig.~\ref{fig:innerLoop} does not 
parallelize well without extra treatment.

\begin{figure}
\begin{minipage}{.3\linewidth}

\begin{verbatim}
 
__kernel
void DCT(__global float * output,
         __global float * input,
         __global float * dct8x8,
         __local  float * inter,
         const    uint    width,
         const    uint    blockWidth,
         const    uint    inverse)
{
    /* ... */
/* parallel_WI_loop { */
    for(uint k=0; k < blockWidth; k++)
    {
        uint index1 = (inverse)? i*blockWidth + k : k * blockWidth + i;
        uint index2 = getIdx(groupIdx, groupIdy, j, k, blockWidth, width);

        acc += dct8x8[index1] * input[index2];
    }
    inter[j*blockWidth + i] = acc;
/* } */
    barrier(CLK_LOCAL_MEM_FENCE);
    /* ... */
}
\end{verbatim}

\end{minipage}

\caption{A kernel with inner loops; a snippet from the DCT kernel of 
the AMD OpenCL SDK code sample suite. Note how the work-item loop 
surrounds the inner-loop which constitutes a parallel region.}
\label{fig:innerLoop}
\end{figure}

The variable loop iteration count makes parallelism extraction hard as the inner loop
cannot be unrolled to increase the number of parallel operations within
one work-item. However, if the inner loop was treated like a loop with a barrier 
inside, the parallelization would be done across the work-items, effectively
leading to a structure as shown in Fig.~\ref{fig:innerLoopHorizontal}.
Thus, the desired end result is a loop interchange between the
inner loop and the work-item loop surrounding that parallel region.

\begin{figure}[!b]
\begin{minipage}{.3\linewidth}

\begin{verbatim}
    for(uint k=0; k < blockWidth; k++)
    {
     /* parallel_WI_loop { */
        uint index1 = (inverse)? i*blockWidth + k : k * blockWidth + i;
        uint index2 = getIdx(groupIdx, groupIdy, j, k, blockWidth, width);

        acc += dct8x8[index1] * input[index2];
     /* } */
    }
\end{verbatim}

\end{minipage}

\caption{A kernel with the inner loop horizontally parallelized.
Here the work-item loop surrounds the inner-loop body, yielding 
a nicely parallelized region.
}
\label{fig:innerLoopHorizontal}
\end{figure}

The legality of this transformation is similar to the legality of
having a barrier inside the loop; all of the work-items
have to iterate the loop the same amount of times. Therefore, additional
\textit{divergence and variable uniformity analysis} is needed in order 
to add such implicit barriers that enforce the horizontal parallelization. 

The uniformity analysis resolves the origin of the variables in the LLVM IR\@. 
The operands of the producer instruction of the variable are recursively
analyzed until a known \textit{uniform} root is found. Uniform variable 
is one that is known to contain the same value for all the work-items
in the work-group. Such a uniform root is usually a constant or a kernel argument.
The uniformity analysis is used to prove that the loop exit condition 
nor the predicates in the path leading to the loop entry do not
depend on the work-item id. That is, the work-item execution does not
\textit{diverge} in such a way that the implicit barrier insertion would
be illegal. Only then the implicit loop barrier is inserted
to enforce the horizontal inner loop parallelization.

\subsection{Handling of Kernel Variables}

Variables of two different scope can be defined in OpenCL C kernel
functions: The per work-item \textit{private} variables and 
the \textit{local} variables which are shared among all 
the work-items in the same work-group. 
While the \textit{private} variables are always allocated in the
OpenCL kernel function definition, there are two ways to allocate \textit{local} 
variables in OpenCL kernels: 
From the host side through the \textit{clSetKernelArg} API (a local
buffer argument in the kernel function), and from the kernel side through 
``automatic local variables'' (variables prefixed in the OpenCL C
description with the \textit{local} address space qualifier). Both
of these cases are handled similarly by pocl by converting the latter case of
automatic locals to an additional work-group function
argument with a fixed allocation size. The additional work-group function
argument for automatic locals is visible in the example kernel of 
Fig.~\ref{fig:kernelCompilationMethods}: A third function argument has been 
added for storing the automatic float array of size four.

What should be noted is 
that \textit{local} data is actually thread-local data from the
point of view of the implementation when multiple work-groups are executed
in parallel in multiple device threads sharing the same physical address
space where the local data is stored. In order to
ensure thread safety, e.g.\ the pthread device driver of pocl
handles all local data by allocating the required local buffers 
in the ``kernel launcher thread'' which calls the work-group function. The
same local space is reused across the possible multiple work-groups
executing in the same device thread.

Private variables, however, need additional processing during the
work-group function generation. As the
original kernel function describes the functionality of a single
work-item, the private variables in the produced multi-WI 
work-group function need to be replicated for all the work-items.
In another point of view, if one considers each work-item to be an independent 
thread of execution, each of the threads must have their own 
separate private context that needs to be used during the execution.
The straightforward way to produce such context space for the work-items is to 
create a \textit{context data array} for each original private variable.
In this array, an element stores the private variable for a single work-item.
Thus, as many elements as there are work-items in the work-group are needed.

Private variables have different life times that affect the
need to store them in a context data array. Some of the private 
variables are used only within one parallel region while some
span multiple regions. In case the lifetime does
not span multiple parallel regions, there is no need to
create a context array for it as the variable is used only during the execution of 
the work-item loop iteration. Such variables can be sometimes 
allocated to registers for their whole lifetime instead of storing them 
to memory.
Fig.~\ref{fig:privateVarLifeSpans} presents the two cases in
a simple kernel which has two parallel regions due to the
barrier in the middle. Variable $a$ is used only in the
first parallel region, thus, it can stay as a scalar within
the produced work-item loop. In contrast, $b$ is used also 
in the latter parallel region and has to be stored in a context
array.
In order to exploit the varying variable lifespans, each private 
variable is examined and if it is used on at least one 
parallel region different from that in which it is defined, a context
array is created. Then, all uses of the variable are replaced
by uses of an element of the newly created array. This analysis is
straightforward in the SSA format; each variable assignment 
defines a new virtual variable of which uses can be found quickly.

\begin{figure}
\begin{center}
\includegraphics[width=0.9\textwidth]{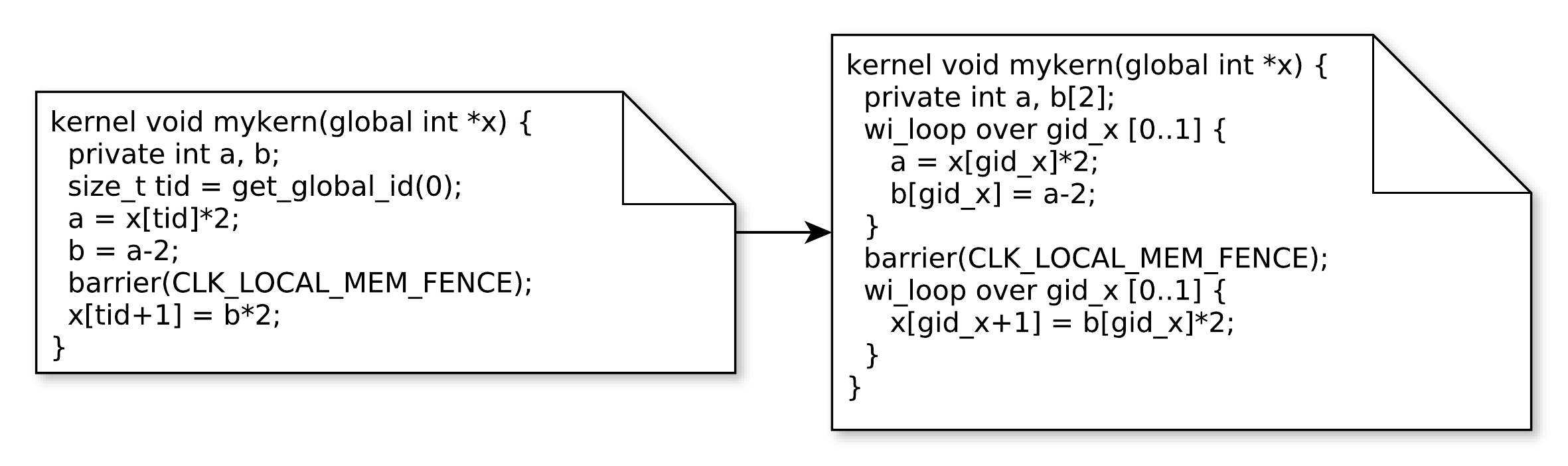}
\caption{Two cases of private variable lifespans: Variable
$a$ is a temporary variable used only in one parallel region, 
while $b$ spans across two parallel regions. The work-group
function generation result is presented in pseudo code for 
clarity; in reality the processing is done on the LLVM
internal representation and the actual variables seen by
the kernel compiler might not match the ones in the input 
kernels due to earlier optimizations.}

\label{fig:privateVarLifeSpans}
\end{center}
\end{figure}

Additional optimization the kernel compiler performs on the private 
variables of the work-group functions is the merging of uniform variables.
The idea is similar to the \textit{Loop-Invariant Code Motion (LICM)~\cite{Aho86}}: 
sometimes the work-items in the work-item loop use variables that are 
invariant, i.e., the value does not change per work-item. In such cases, 
context data space can be saved by merging the variables to a single scalar 
variable that is shared across the work-items. If this is left to 
a later LICM optimization on the work-item loop, it might not succeed 
due to the need to analyze the accesses to the context array locations to 
prove the values are the same. 

The kernel compiler uses the same uniformity analysis as was described 
in Section~\ref{sec:horizontalLoopParallelization} to detect and merge
such variables. In some cases this optimization is counter-productive 
in case it leads to the need to broadcast values across the lanes of 
SIMD-based machines, which might be expensive. In that case it 
can be more efficient to also replicate the uniform values just 
to avoid the communication costs. Taking advantage of this 
machine-specific property is left for future work.

\section{Vectorized Mathematical Library Functions}
\label{sec:MathFunctions}

OpenCL extends the usual mathematical elemental library functions found in C
(e.g.\ sin, cos, sqrt) to accept vector arguments as well. To achieve
good performance for computationally-bound kernels, efficient, 
vectorized implementations for these are needed.
%
We designed \textit{Vecmathlib}~\cite{web:vecmathlib} as a pocl
sub-system to address
this need. Vecmathlib provides efficient, accurate, tunable, and most
importantly vectorized mathematical library functions. 
It seeks to design new algorithms for calculating elemental functions
that execute efficiently when interspersed with other application
code. This is in contrast to many other libraries, such as e.g. IBM's
ESSL or Intel's VML, which are designed to be called with arrays of
many (thousands) of elements at once.

Vecmathlib is implemented in C++, and intended to be called on SIMD
vectors, e.g.\ those provided by SSE or AVX instruction sets, or
available on ARM, Power7, and Blue Gene architectures. The same algorithms
also work efficiently on accelerators such as GPUs. Even for scalar
code, Vecmathlib's algorithms are efficient on standard CPUs.

Vecmathlib consists of several components:
\begin{itemize}
\item Type traits, defining properties of the available floating-point
  types (such as half, float, double) and their integer equivalents
  (short, int, long), extending \texttt{std::numeric\_limits};
\item Templates for SIMD vector types over these floating-point types,
  called \texttt{realvec<typename T, int D>};
\item Generic algorithms implementing mathematical functions; these
  algorithms act on SIMD vectors to ensure they are efficiently
  vectorized;
\item Particular vector type definitions depending on the system
  architecture, providing e.g.\ \texttt{realvec<double,2>} if Intel's
  SSE2 instructions are available. These definitions use efficient
  intrinsics (aka machine instructions) if available, or else fall
  back to a generic algorithm.
\end{itemize}
Thus Vecmathlib directly provides efficient vector types for those
vector sizes that are supported by the hardware. Other vector sizes
are then implemented based on these, so that
e.g.\ \texttt{realvec<float,2>} may be implemented via extension to
\texttt{realvec<float,4>} (with two unused vector elements), or
\texttt{realvec<float,8>} operations may be split into two
\texttt{realvec<float,4>} if necessary. This happens transparently, so
that OpenCL's types \texttt{float2} or \texttt{float8} have their
expected properties.

\subsection{Implementation}


Low-level mathematical functions such as fabs, isnan, or signbit are
implemented via bit manipulation. These algorithms currently assume
that floating point numbers use the IEEE layout
\cite{IEEE-Std-754-2008,Goldberg}, which
happens to be the case on all modern floating-point architectures. For
example, fabs is implemented by setting the sign bit to 0.

Mathematical functions where the inverse can be calculated
efficiently, such as reciprocal or square root (where the inverses can
be determined via a simple multiplication), are implemented via
calculating an initial guess followed by an iterative procedure. For
example, $\mathrm{sqrt}(x)$ is implemented by first dividing the
exponent by two via an integer shift operation, and then employing
Newton's root finding method~\cite{Press2007}
via iterating
$ r_{n+1} := ( r_n + x/r_n ) / 2 $
where $r_n$ is the current approximation. This algorithm doubles the
number of accurate digits with every iteration. 

Most mathematical functions, however, are calculated via a range
reduction followed by a polynomical expansion. For example, $\sin(x)$
is calculated by first reducing the argument $x$ to the range $\lb0;
2\pi)$ via the sine function's periodicity, then reducing the range
further to $[0; \pi/2]$ via the sine function's symmetries, and
finally expanding $\sin(x)$ into Chebyshev polynomials~\cite{Press2007} that minimize
the maximum error in this range~\cite{Muller2006}.

\subsection{Vectorizing Scalar Code}
\label{sec:VectorizingScalarCode}

Instead of implementing \emph{vectorized} mathematical functions (that
take vector arguments), it would be advantageous to implement
\emph{vectorizable} functions (that take scalar arguments), and which
would then automatically be vectorized by the compiler.
For example, the SLEEF library \cite{Shibata2010,Shibata2013} takes
this approach. This would
certainly simplify the implementation of Vecmathlib itself. However,
this is unfortunately not possible for the following reason:
the high-level algorithms depend on low-level functions such as
e.g.\ fabs, floor, or signbit. Whether these low-level functions are
provided efficiently by vector hardware, or whether they need to be
calculated via bit manipulation, is architecture dependent. 
%
We assume that LLVM's vectorizer will in the future be able to
vectorize such calls. The logic required for this is exactly the logic
already found in Vecmathlib, so one obvious way to implement this
functionality in LLVM is via utilizing Vecmathlib.

\section{Performance Evaluation}
\label{sec:Evaluation}

For evaluating the current performance of the proposed approach implemented in 
pocl, we used the suite of example applications available
in the \textit{AMD Accelerated Parallel Processing Software Development 
Kit}~\cite{AMD-APP-SDK}. The example
applications in the AMD APP SDK suite allow timing the execution and to iterate
the benchmark multiple times. Multiple execution iterations are used to reduce
cache effects to numbers and to allow the kernel compilers to amortize the 
kernel compilation time across kernel executions.


The benchmark suite was executed on various platforms supported by pocl. 
The same unmodified benchmark suite was also executed using the best found vendor 
implementation of OpenCL for the platform at hand for giving an idea where 
the performance is at in comparison to the most commonly used implementations.
It should be noted that this version of the benchmark has been optimized
for previous generation AMD GPUs with VLIW lanes. For example, many of the 
cases use explicit vector code which has to be scalarized by the pocl kernel
compiler for more efficient horizontal work-group vectorization. 

The processors in the tested platforms and their available parallel computation resource types 
are summarized in Table~\ref{table:benchmarkPlatforms}. Pocl framework exploits the parallel 
resources as follows:
a) thread-level parallelism (TLP); multiple work-groups in multiple hardware threads or cores,
b) instruction-level parallelism (ILP); dynamic or static multi-issue cores enable concurrent 
execution of multiple operations from each parallel region (from the same work-item or from
multiple work-items), and c) data-level parallelism (DLP); SIMD instruction sets allow 
executing either intra-kernel vector instructions directly or lock-step executing 
matching operations from multiple work-items.

\begin{table}
 
\begin{center}
\begin{tabular}{l|c|c|c}
                 & TLP       & ILP          & DLP             \\
\hline
Intel Corei7~\cite{IntelHaswell} & 4 cores, 2 threads each   & 8-issue
out-of-order & AVX2 (8 float, 4 double) \\
ARM Cortex-A9~\cite{ARMCortexA9} & 2 cores & out-of-order & NEON \\
Power Processor Element & 2 threads & 2-issue in-order & AltiVec \\
TTA & n/a & static multi-issue & n/a \\
\end{tabular}
\caption{Different types of parallel computation resources exploited in the tested platforms.
The resources are categorized to the type of parallelism they serve:
thread-level parallelism (TLP), instruction-level parallelism (ILP), 
and data-level parallelism (DLP).
}

\label{table:benchmarkPlatforms}
\end{center}

\end{table}

\subsection{Intel x86-64}

The first evaluated platform is the most popular instruction set architecture used in
current personal computers and work stations, the Intel 64bit x86 architecture. 
For benchmarking this platform 
we used a workstation with an \textit{Intel Core i7-4770}
CPU clocked at 3.4~GHz. The workstation 
had 16~GB of RAM and ran the Ubuntu Linux 12.04 operating system. The kernel execution time 
performance results are given in Fig.~\ref{fig:resultsIntelLinux}. There
were two proprietary OpenCL implementations on the platform we could compare
against, one from AMD and another
from Intel. 
This benchmark set indicates great performance can be achieved using pocl 
despite the fact that there are several performance opportunities that are under implementation.
For several of the benchmark applications pocl already outperforms the available proprietary
implementations. However, a few bad results stick out from the results: 
\textit{BinarySearch} and \textit{NBody}. 
We analyzed the cases and listed the additional optimizations that should 
help to reach the vendor implementation performance also for these cases. 
They are discussed in the Conclusions and Future Work section
later.

\begin{figure}
\begin{center}
\includegraphics[width=0.8\textwidth]{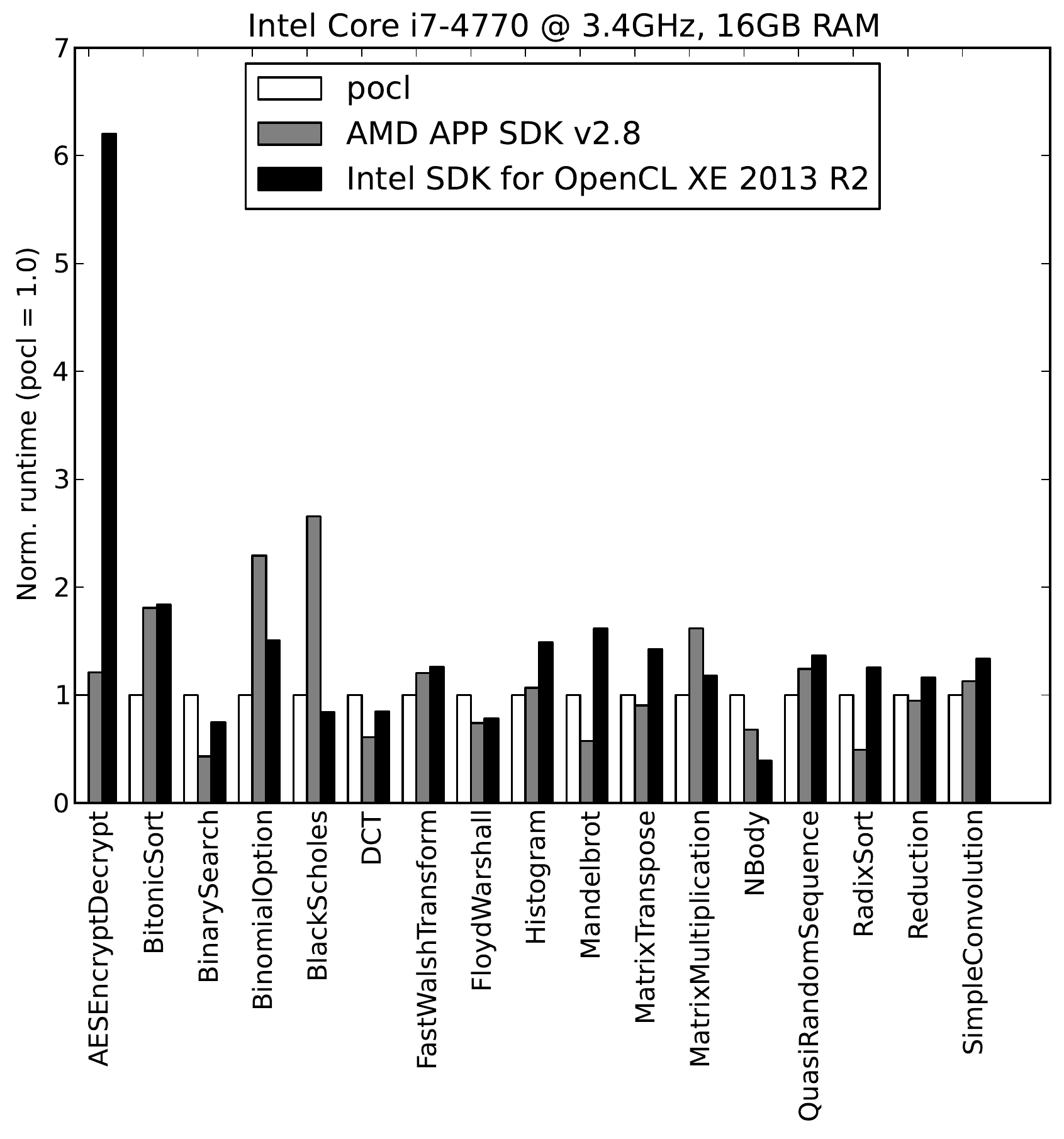}
\caption{Benchmark execution times (smaller is better) with Intel Core i7, 
in a Linux environment. The benchmark runtime achieved with pocl is 
compared to the two available proprietary implementations for the platform.}
\label{fig:resultsIntelLinux}
\end{center}
\end{figure}

\subsection{ARM Cortex-A9}

ARM CPUs are currently the standard choice for general purpose processing in 
mobile devices. We benchmarked the
ARM platform using the PandaBoard with Ubuntu Linux 12.04 installed as the operating
system. The PandaBoard has an \textit{ARM Cortex-A9}~\cite{ARMCortexA9} CPU
which is an out-of-order multiscalar architecture with a NEON~\cite{ARM-NEON} SIMD unit.
The CPU is clocked at 1~GHz, and the platform has 1~GB of RAM.
On this platform we could not compare against a vendor supplied OpenCL implementation 
as ARM does not supply one (as of February, 2013) for their CPUs, but only for 
their Mali GPUs. 
Benchmarking results against FreeOCL~\cite{FreeOCL} (albeit it is not a
performance-oriented implementation) are shown in Fig.~\ref{fig:resultsCortexA9}.
The BinomialOption test case failed to work with FreeOCL.

\begin{figure}[t]
\begin{center}
\includegraphics[width=0.8\textwidth]{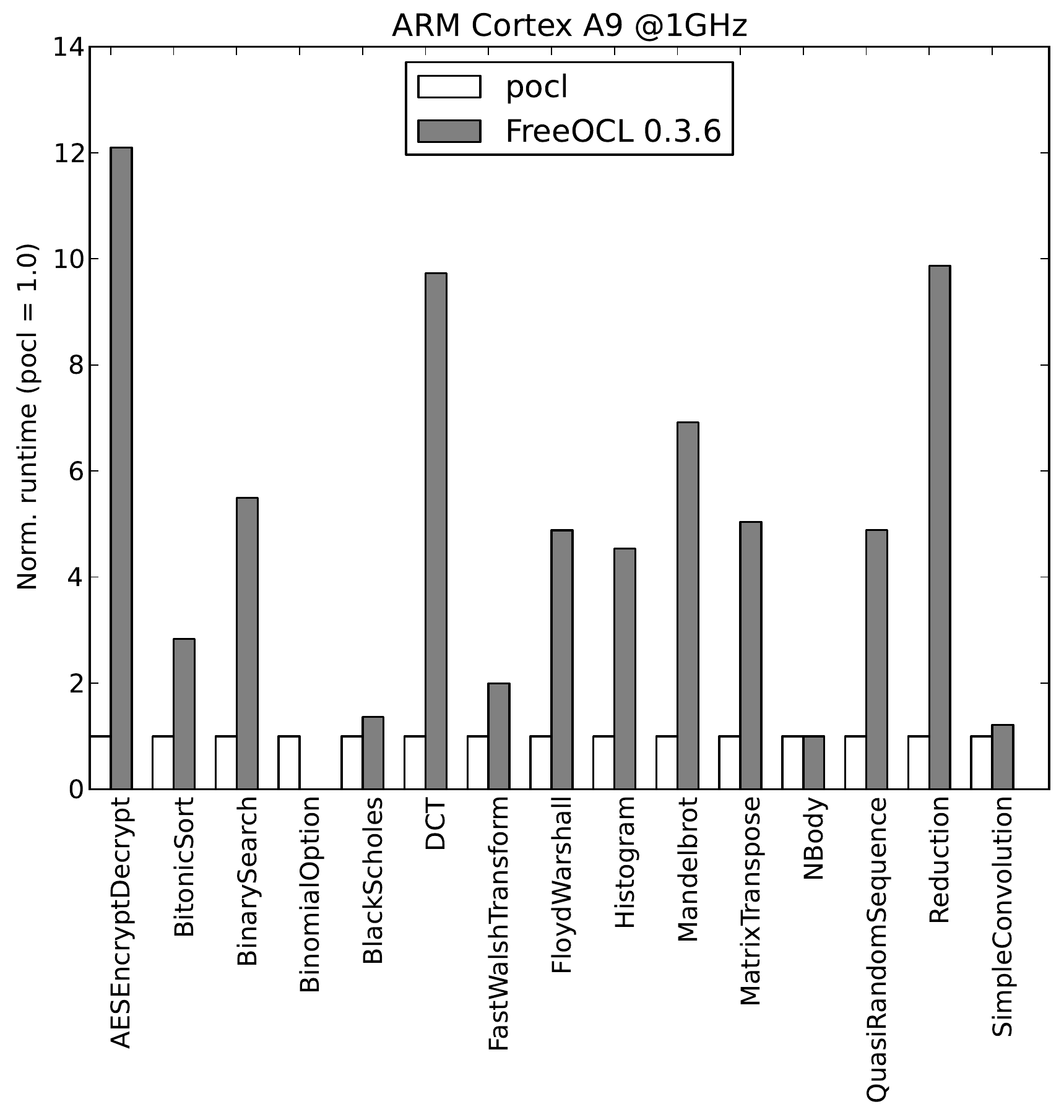}
\caption{Benchmark execution times (smaller is better) with ARM Cortex-A9, 
1GB RAM in Ubuntu Linux.}
\label{fig:resultsCortexA9}
\end{center}
\end{figure}

\subsection{STI Cell Broadband Engine / Power Processing Element}
 
The Cell Broadband Engine~\cite{SPUALT} is a heterogeneous multiprocessor 
consisting of a PowerPC and 8 Synergistic Processing Units (SPU's). pocl can utilize 
the PowerPC via the \emph{basic} and \emph{pthreads} drivers.
The \emph{spu} driver in pocl can execute programs on the SPU processors. However, a majority of the test cases
failed with compiler errors due to the immature state of the LLVM SPU backend. Also, as LLVM has removed the
SPU backend since the 3.2 release, the benchmarks were not run on the SPU parts of the Cell.
The PowerPC of the cell was benchmarked on a Sony Playstation 3, running the Debian \emph{sid} operating system.
The IBM OpenCL Development Kit v0.3~\cite{IBMOCL} was used as a benchmark reference on this platform.
The reference benchmarks were run using the 'CPU'-device in both the OpenCL implementations, i.e. the SPUs were 
not used. The comparative results varied significantly (see Fig.~\ref{fig:resultsPPC}) with pocl
performing the best in the vast majority of the benchmarks.

\begin{figure}
\begin{center}
\includegraphics[width=0.8\textwidth]{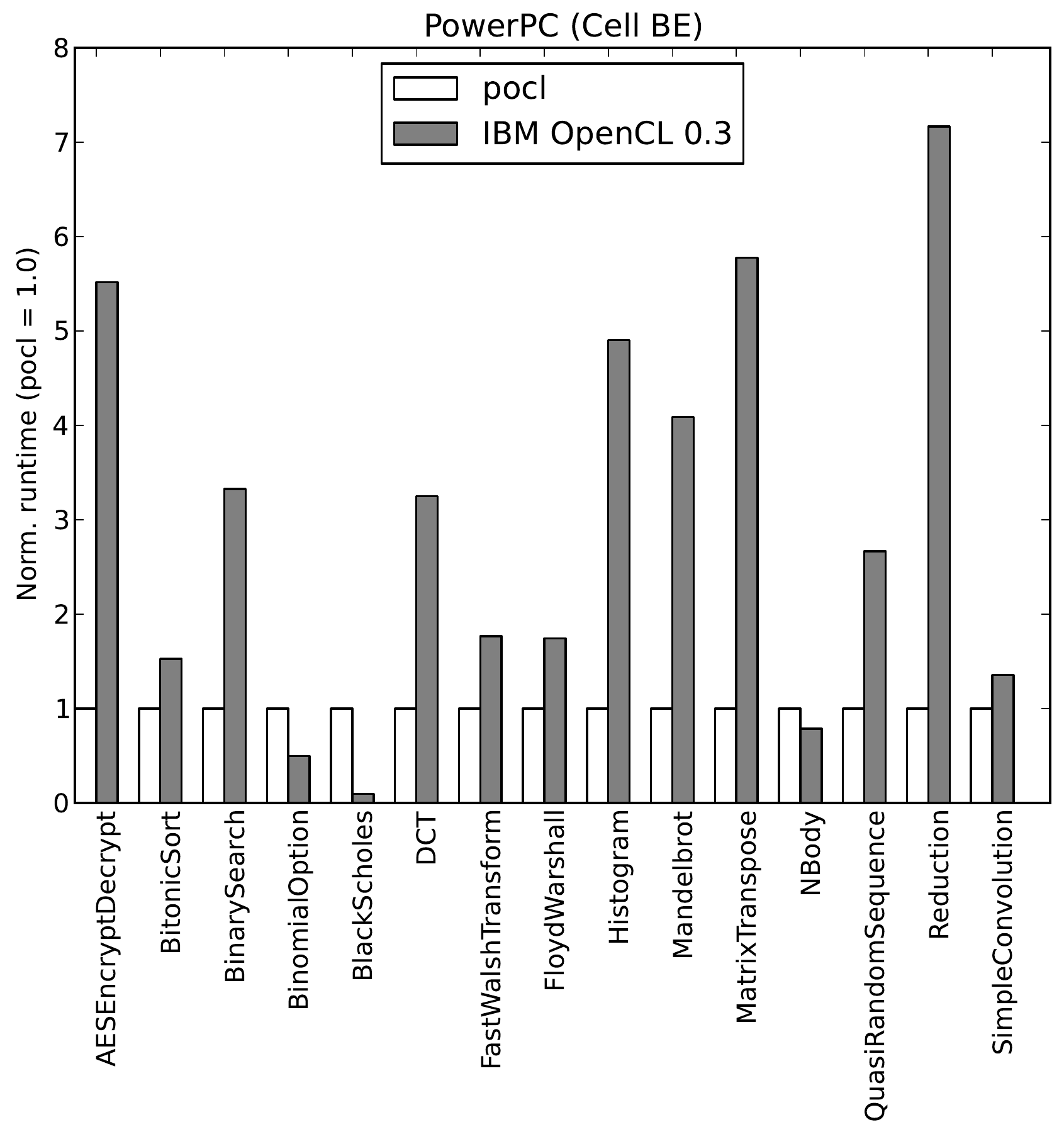}
\caption{Benchmark execution times (smaller is better) with STI CellBE @ 3.2GHz, 256MB RAM running Debian sid.
Both OpenCL implementations utilize the PowerPC processor only.}
\label{fig:resultsPPC}
\end{center}
\end{figure}

%
%
%

\subsection{Static multi-issue}

An important feature of the pocl kernel compiler is its separation of 
parallelism exposing transformations (parallel region formation) and the
actual parallelization of the known-parallel regions to the processor's
resources. 
The platforms in the previous benchmarks have exploited the 
parallelism available in dynamic multi-issue CPUs, their SIMD 
extensions and multiple cores or hardware threads.
Static multi-issue architectures are interesting especially in low 
power devices as they reduce 
the hardware logic needed to support parallel computation and they rely on 
the compiler to exploit the parallel function units in the 
machine~\cite{201766}. 
%
In order to test how well the proposed kernel
compilation techniques can exploit the parallelism in VLIW-style machines, we designed 
a \textit{Transport Triggered Architecture (TTA)} processor with multiple parallel function 
units. For this, the publicly available processor design toolset \textit{TCE (TTA-Based 
Co-design Environment)}~\cite{TCEFPGA-ALT,TCEURL} was used.

\textit{Transport Triggered Architecture (TTA)} is a VLIW
architecture with a programmer exposed interconnection network~\cite{Corporaal:97:TTA_BOOK}. 
It exposes the instruction level parallelism statically like the traditional VLIWs but 
adds more instruction scheduling freedom due to the transport programming model. 
For this benchmark we used a pocl device layer implementation that accesses the instruction 
set simulator engine of TCE for modeling a TTA-based accelerator device. The simulator engine
is instruction cycle count accurate, thus allows measuring the scalability of scheduling 
the multi-WI work-group functions statically to function units. The processing
resources in the designed TTA are listed in Table~\ref{table:ttaResources}.

\begin{table}[h]

\begin{center}

\begin{tabular}{|l|c|}
\hline
Resource & \# \\
\hline
Integer register files (1rd+1wr port, 32 regs each) & 4 \\
Boolean register files (1rd+1wr port, 16 regs each) & 5 \\
Integer ALUs & 4 \\
Float add+sub units & 4 \\
Float multiplier units & 4 \\
Load-store units (for global and local) & 9 \\
\hline
\end{tabular}

\caption{Computational resources in the TTA datapath used in the ILP benchmark.
}
\label{table:ttaResources}
\end{center}
\end{table}

The test application used here was the unmodified
DCT benchmark from the AMD SDK. 
This benchmark is a good example which benefits from the inner loop 
horizontal parallelization of the kernel compiler to improve the
exploitable parallelism (see Section~\ref{sec:horizontalLoopParallelization}). 
The kernel has two loops without barriers and have an iteration count 
given as the kernel argument. Thus, without the horizontal parallelization
transformation, the loops could not be unrolled and are executed for each
work-item in a sequence with very limited instruction-level parallelism. 
The kernel execution time without the horizontal
parallelization was 53.5~ms and 10.2~ms (scaled to 100~MHz) when the 
horizontal inner loop parallelization pass was used. Thus, the ILP increase when
exploiting the kernel parallelization pass was roughly five-fold.

\subsection{Performance of the Built-in Functions}

We evaluated the speed of certain mathematical functions implemented
using the Vecmathlib for various
vector sizes on an Intel Core i7 (with SSE4.2 vector instructions) and
on a PS3 (with Altivec vector instructions). We compared
scalarizing the
function calls and marshalling them to libm, which presumably provides
an optimized scalar calculation, to Vecmathlib's vectorized
implementation. Results are presented in tables
\ref{tab:vecmathlib-benchmark} and
\ref{tab:vecmathlib-benchmark-altivec}. The benchmarks used the
\texttt{-ffast-math} option, and each calculation was repeated
10,000,000 times in a loop to obtain more accurate measurements.

\begin{table}
  \begin{tabular}{|lrl|r|rrr|}\hline
  type   & vector & impl. & overhead &      exp &      sin &     sqrt \\
         &   size &       & [cycles] & [cycles] & [cycles] & [cycles] \\\hline
  float  &      1 & libm  &      2.6 &     38.3 &     47.5 &     13.2 \\
         &        & SSE2  &      2.3 &     12.9 &     12.4 &     12.1 \\
         &      4 & libm  &     27.5 &    433.6 &    474.3 &     92.0 \\
         &        & SSE2  &      9.3 &     52.4 &     41.6 &     21.8 \\
  double &      1 & libm  &      2.3 &     23.9 &     33.2 &     17.5 \\
         &        & SSE2  &      2.3 &     23.1 &     21.6 &     18.2 \\
         &      2 & libm  &      6.2 &     92.2 &    127.1 &     51.8 \\
         &        & SSE2  &      4.6 &     52.5 &     41.8 &     21.9 \\\hline
  \end{tabular}
  \caption{Performance benchmark results, showing execution time in
    cycles (lower is better). This compares a naive, scalarizing
    implementation via libm to Vecmathlib for exp, sin, and sqrt. The
    column ``overhead'' shows the approximate overhead of the
    benchmarking harness. Note that scalarization by itself is
    expensive since it requires vector shuffle operations (see
    overhead column). Note also that, in many cases, Vecmathlib is
    more efficient than calling libm even in the scalar case. On this
    system, exp and sin are implemented via a generic algorithm,
    whereas sqrt is implemented via a machine instruction.}
  \label{tab:vecmathlib-benchmark}
\end{table}

\begin{table}
  \begin{tabular}{|lrl|r|rrr|}\hline
  type   & vector & impl.   & overhead &      exp &      sin &     sqrt \\
         &   size &         & [cycles] & [cycles] & [cycles] & [cycles] \\\hline
  float  &      1 & libm    &      0.3 &     30.1 &     18.4 &      6.9 \\
         &      4 & libm    &      1.4 &    485.6 &    302.6 &    114.1 \\
         &        & Altivec &      1.3 &     13.4 &     33.9 &      8.7 \\\hline
  \end{tabular}
  \caption{Performance benchmark results, showing execution time in
    cycles (lower is better). The qualitative results are similar to
    Table~\ref{tab:vecmathlib-benchmark}. On this system, exp and sin
    are implemented via a generic algorithm, whereas sqrt profits from
    a special machine instruction.}
  \label{tab:vecmathlib-benchmark-altivec}
\end{table}

It is clearly evident that Vecmathlib's implementation is in all cases
at least as efficient as libm, even in the scalar case. For vector
types, Vecmathlib is always significantly more efficient, since
scalarizing (disassembling and later re-assembling) a vector is an
expensive operation in itself. In particular for single precision,
Vecmathlib is significantly faster (for exp and sin) than libm; this
is presumably because libm's implementation uses the Intel
\texttt{fexp} and \texttt{fsin} machine instructions which always uses
double precision, whereas Vecmathlib evaluates these functions only
for single precision. For the scalar sqrt function, there is almost no
speed difference, because both libm and Vecmathlib employ the SSE2
\texttt{sqrtss} instruction.

\section{Related work}
\label{sec:RelatedWork}


There has been previous work related the kernel compiler transformations. For example, 
\textit{Whole Function Vectorization (WFV)}~\cite{WFV,ImprovingOpenCLPerf} 
 is a set of vectorization techniques tailored 
for efficient vectorization of SPMD descriptions  
such as the OpenCL work-group functions with multiple work-items.
Similar approach is used in the \textit{Implicit Vectorization Module} of
Intel's OpenCL SDK~\cite{IntelOpenCLVec}. These solutions rely on 
a certain type of parallel computation 
resources during the kernel compilation such as vector instruction set extensions and
perform the vectorization by expanding the scalar operations to their vector
counterparts, whenever possible.
This style of ``monolithic approaches'' are platform specific with limited support for 
performance portability. That is, when adding support to new devices with different
parallel hardware, larger part of the kernel compiler has to be updated. The approach 
taken by pocl is to split the work-group vectorization to a generic step that 
identifies the parallel regions, converts the regions to data parallel loops, and 
retains the parallelism information to the later stages of compilation. The 
later stages are the same as in a standard vectorizing compilation. Therefore,
when porting pocl to a new platform that supports LLVM, minimal effort is needed 
to get a working and efficient implementation.

Extracting parallelism from sequential programs, especially 
from loops that are not (known to be) parallel is a challenge that has 
received extensive attention in the past decades. As an example of one of the more 
recent works, the techniques proposed by \textit{Nicolau et al.}
enhance the thread level parallelism of loops with inter-iteration dependencies by intelligent 
placement of \textit{post and wait} synchronization primitives~\cite{Nicolau1,Nicolau2}. 
In the case of the compilation flow presented in this article, the key problem is not 
the extraction of parallelism from a serial program because the input kernels 
are parallel by default and explicitly synchronized by barriers. The 
complexity is in finding at compile time the parallel regions of multiple instances 
of the work-item descriptions,
and this parallelism is typically mapped to finer grain resources such as
vector lanes or function units. However, the work of \textit{Nicolau et al.} could be 
used to enlarge the found parallel regions to make their execution more efficient 
in machines that execute work-items from a single work-group in separate threads or 
cores.

The earliest mentions we found of the idea of generating ``work-group functions'' 
that execute multiple work-items in parallel to improve the performance of 
SPMD optimized programs on non-SPMD hardware has been published previously in the 
context of CUDA~\cite{MCUDA}. The same idea is referred to as ``work-item
coalescing'' for OpenCL in~\cite{OpenCLFW}. The MCUDA ideas can be applied directly to
OpenCL kernels as the concept of SPMD descriptions with barrier synchronization is 
identical in both languages. These previous works 
present the key idea of identifying regions that are parallel and executing them for 
all work-items. 
The previous works are implemented as source-to-source transformations which
is great for portability, but lacks the performance portability benefits.
The feature that impacts the \textit{performance portability} aspect 
the most is the transfer of parallelism information. In case of source to source 
approaches, the information is lost as the parallel regions are converted to
serial loops, ending up with the usual alias analysis complexity present in,
e.g. C loops. The transferring of parallel loop data using the LLVM IR metadata 
enables pocl to maintain the information of the data parallel regions in
work-group and benefit the later optimization stages.
We also perform additional 
parallelism-improving optimizations such as inner-loop parallelization of kernels 
without barriers by selectively converting them to kernels with implicit barriers,
effectively parallelizing the outer loop (work-item loop). Finally, a major 
drawback of the source-based approaches is the language-dependence. With the 
introduction of the SPIR standard~\cite{SPIR12}, it is now possible to define 
OpenCL kernels using multiple alternative languages. Because SPIR uses LLVM IR, 
the proposed kernel parallelization techniques apply to kernels loaded from 
SPIR binaries as well.

There are also previous attempts to provide portable 
OpenCL implementations.
One of the well known ones is Clover~\cite{Cloverhome}, which is an
OpenCL implementation providing GPU
computation support using open source drivers. 
%
%
Clover implements the work-group
barriers using light weight threads (or ``fibers``). 
A similar fiber-based approach is  \textit{Twin Peaks}~\cite{TwinPeaks,SetjmpLongjmpBarrier}, 
which proposes using optimized setjmp/longjmp functions for implementation.
The drawback with the fiber approach is that the light weight threads do not allow 
implicit static parallelization of multi-WI work-groups~\cite{CloverBarriers}.
Therefore, the performance portability and ''scaling`` is limited with
these solutions.
After all, 
the main source for parallelism in OpenCL kernels is the ability to 
execute operations from multiple work-items in any order, also statically using
fine grained parallel resources such as SIMD or VLIW instructions. 
This cannot be achieved when threads with independent control are 
spawned for work-items. There are also overheads in the fiber approach
due to the context switches itself, but it is clear that the 
capability to horizontally parallelize work-groups has the main performance 
benefit in the proposed work.

FreeOCL~\cite{FreeOCL} is an open source implementation of the OpenCL 1.2. The 
target of FreeOCL is stated as \textit{``It aims to provide a debugging tool and a reliable 
platform which can run everywhere.''} FreeOCL relies on an external C++ compiler
to provide a platform portable implementation, but again does not provide a kernel
compiler with static parallelization of work-items to improve \textit{performance} 
portability. Like several other implementations, it relies on the fiber approach for 
implementing multi-WI work-group execution.

The proposed approach attempts to improve the performance portability over a wide range of 
platforms; the pocl kernel compilation does not rely on any specific parallel computation 
resources (unlike WFV, which relies on vectorization). This is
apparent in its separation of the compiler analysis that expose the parallel 
regions between work-group barriers from the generic parallelization
passes (such as vectorization or VLIW scheduling).
This style of modularized kernel compilation improves the performance 
portability of the OpenCL implementation thanks to the freedom to map the parallel operations
in the best way possible to the resources of the device at hand. 

The proposed solution uses static program
analysis to avoid using threads with independent control flow
for executing multiple work-item kernels with barriers, which allows improved performance 
portability compared to fiber-based approaches like Clover and Twin Peaks.  
For improving the platform portability aspect, the proposed solution uses only C language 
for the the host API implementation (instead of C++ as used in Clover) 
in order to allow porting the code to a wider range of embedded platforms 
without extensive compiler or runtime support.

\section{Conclusions and Future Work}
\label{sec:Conclusions}

In this article, we described a modular performance portable OpenCL kernel compiler 
and a portable OpenCL implementation called \textit{pocl}. 
The modular kernel compiler provides an efficient 
basis for kernel compilation on various devices with parallel resources of different 
granularity. The kernel compiler is constructed to separate the analysis that expose 
the parallelism from multi-WI work-groups and to more standard optimizations that 
perform the actual static parallelization of the parallel regions to different 
styles of fine-grained parallel computation hardware, such as SIMD, VLIW, or 
superscalar architectures. The data parallelism information of multiple work-item
work-group functions is transferred using LLVM IR metadata for later compilation 
phases.
The experiments on different processor architectures showed that 
pocl can be used to port OpenCL applications efficiently and it can exploit 
various kinds of parallelism available in the underlying hardware.
The pocl framework can also be used as an experimentation platform for the popular 
OpenCL programming standard, and it provides an OpenCL implementation framework 
for engineers designing new parallel computing devices.

The pocl kernel compiler itself is fully functional and usually very efficient. 
It was shown that most of the benchmarked applications were faster or close to
as fast as the best proprietary OpenCL implementation for the platform at hand.

At its current state, most of the performance improvements to the kernel 
compiler of pocl will be language generic in nature. They can be implemented to 
the LLVM infrastructure and as a result benefit also non-OpenCL programs. 

For example, we plan to add selective scalarization of vector code inside loops. 
That is, in case the loop vectorization cannot be applied for some reason, the 
original vector code added by the programmer should be left
intact to still allow exploiting some SIMD instructions. Same applies to the 
aggressive inlining of built-ins and other functions. The current way of 
inlining everything to the kernel function can be counter-productive 
due to the larger instruction cache footprint in case it does not improve the 
vectorization or other form of static parallelization of the work-items. 
This will be more the case in the future as larger and larger 
kernels are implemented using OpenCL. We plan to more intelligently choose 
when to inline and when not on work-item loop basis. A method similar to the 
one presented in~\cite{Cammarota:2013:DIV:2450247.2450262} could be used. 
All of the worst-performing cases presented in Section~\ref{sec:Evaluation} would 
benefit from these. 

Another 
bottleneck we identified is the limited support for if-conversion~\cite{p177-allen} in 
the current LLVM version. The inability to predicate some otherwise statically 
parallelizable work-item loops is one of the biggest slowdowns in the worst 
performing benchmark cases. Related to this, there are 
several OpenCL-specific optimizations we plan to experiment with. For 
example, improving the parallelization of kernels with diverging 
branches (parts executed only by a subset of the work-items) is one 
of the low-hanging fruits. There is some previous work available that 
is targeted towards enhanced load-balancing which could be adapted 
to improving the fine-grained parallelization on machines with
limited support for predication as well~\cite{Kejariwal}.

%

%

\begin{acknowledgements}
  The work has been financially supported by the Academy of Finland (funding 
  decision 253087), Finnish Funding Agency for Technology and Innovation
  (project "Parallel Acceleration", funding decision 40115/13), 
  by NSF awards 0905046, 0941653, and 1212401, as well as an NSERC grant to 
  E. Schnetter. The authors would also like to thank the constructive
  comments and references pointed out by the reviewers.
\end{acknowledgements}

\bibliographystyle{spmpsci} 
\bibliography{Bibliography}

\end{document}